\documentclass{article}
\usepackage{blindtext}
\usepackage[utf8]{inputenc}
\usepackage{float}
\usepackage{bbm}
\usepackage{bm}
\usepackage{crossreftools}

\usepackage{hdb_macros}
\usepackage{fullpage}

\DeclareMathOperator{\conv}{conv}

\usepackage[multiuser,inline,nomargin]{fixme}
\FXRegisterAuthor{a}{ea}{\color{purple}Amir}
\FXRegisterAuthor{h}{eh}{\color{blue}Huck}
\FXRegisterAuthor{w}{ew}{\color{teal}Willow}
\fxusetheme{color}

\newcommand{\EMPH}[1]{\textit{\textbf{#1}}}
\newcommand{\1}{\vec{e}}
\newcommand{\im}{\text{im}}
\DeclareMathOperator{\svol}{svol}

\usepackage{tikz}

\newcommand{\W}{\mathcal{W}}
\renewcommand{\S}{\mathcal{S}}
\renewcommand{\C}{\mathcal{C}}
\newcommand{\V}{\mathcal{V}}
\newcommand{\T}{\mathcal{T}}

\makeatletter
\newcommand{\myitem}[1]{%
\item[#1]\protected@edef\@currentlabel{#1}%
}
\makeatother

\title{Topological $k$-metrics}
\date{\today}
\author{Willow Barkan-Vered\thanks{Oregon State University. \email{barkanvt@oregonstate.edu}. This author's research is supported by NSF grants CCF-1941086 and CCF-1816442.} \and Huck Bennett\thanks{Oregon State University. \email{huckbennett@gmail.com}.} \and Amir Nayyeri\thanks{Oregon State University. \email{nayyeria@oregonstate.edu}. This author's research is supported by NSF grants CCF-1941086, CCF-1816442, and CCF-2311179.}}

\begin{document}

\maketitle
\listoffixmes

\begin{abstract}
Metric spaces $(X, d)$ are ubiquitous objects in mathematics and computer science that allow for capturing (pairwise) distance relationships $d(x, y)$ between points $x, y \in X$. Because of this, it is natural to ask what useful generalizations there are of metric spaces for capturing ``$k$-wise distance relationships'' $d(x_1, \ldots, x_k)$ among points $x_1, \ldots, x_k \in X$ for $k > 2$.
To that end, G\"{a}hler (Math. Nachr., 1963) (and perhaps others even earlier) defined \emph{$k$-metric spaces}, which generalize metric spaces, and most notably generalize the triangle inequality $d(x_1, x_2) \leq d(x_1, y) + d(y, x_2)$ to the ``simplex inequality'' $d(x_1, \ldots, x_k) \leq \sum_{i=1}^k d(x_1, \ldots, x_{i-1}, y, x_{i+1}, \ldots, x_k)$. 
(The definition holds for any fixed $k \geq 2$, and a $2$-metric space is just a (standard) metric space.)

In this work, we introduce \emph{strong $k$-metric spaces}, $k$-metric spaces that satisfy a topological condition stronger than the simplex inequality, which makes them ``behave nicely.''
We also introduce \emph{coboundary $k$-metrics}, which generalize $\ell_p$ metrics (and in fact all finite metric spaces induced by norms) and \emph{minimum bounding chain $k$-metrics}, which generalize shortest path metrics (and capture all strong $k$-metrics).
Using these definitions, we prove analogs of a number of fundamental results about embedding finite metric spaces including Fr\'{e}chet embedding (isometric embedding into $\ell_{\infty}$) and isometric embedding of all tree metrics into $\ell_1$. 
We also study relationships between families of (strong) $k$-metrics, and show that natural quantities, like simplex volume, are strong $k$-metrics.

\end{abstract}

\newpage

\section{Introduction}
\label{sec:introduction}

Metric spaces $(X, d)$ consist of a set of points $X$ and a metric function $d$ that specifies the pairwise distances of elements in $X$.
Metric spaces capture and abstract most familiar notions of distance, such as the $\ell_p$ distance between pairs of points in $\R^m$ and the length of the shortest path between pairs of vertices in a weighted graph. %
A major line of research has studied \emph{metric embeddings}, which work to relate different families of metrics and classify them according to their ``richness.'' Formally, it seeks to construct isometric or low-distortion embeddings between different families of metric spaces.
Well-known examples of such results are Fr\'{e}chet embedding into $\ell_\infty$~\cite{Frechet1906SurQuelques}, Bourgain's theorem for embedding into $\ell_1$~\cite{Bou85}, the Johnson-Lindenstrauss lemma for dimension reduction in $\ell_2$~\cite{JL84}, the embeddability of tree metrics into $\ell_1$~\cite{LinialMagenSaks1998TreesEucMetrics}, and Bartal's theorem for embedding into a distribution of tree metrics~\cite{Bartal96IntoProbTrees}. See Matou\v{s}ek~\cite{matousek13} for a survey. 
Besides being inherently mathematically interesting, these results have many applications in computer science, including most notably in the design of approximation algorithms for problems involving flows and cuts, as well as for problems in network design. See, e.g.,~\cite{AvisDeza1991CutCone,Linial94, Bartal96IntoProbTrees,LeightonRao1999MultcomMaxFlow,journals/jacm/AroraRV09}.

\paragraph{Generalizing to larger $k$.}
Metrics capture pairwise distance relationships between points $x, y$ in a set $X$, but do not (directly) capture relationships among more points.
Moreover, notions of ``distance'' that capture relationships between $k > 2$ points have been studied less extensively, and much less is known about their structural properties and possible applications.
In fact, it is not even a priori clear what the right way to generalize metric spaces to $k > 2$ points is.
Potentially the most notable and natural generalization of metric spaces is \emph{$k$-metric spaces}, which were apparently introduced by \cite{Gahler-63}.%
\footnote{We note that some other work on $k$-metrics---including~\cite{DezaRosenberg2000Semimetrics} and~\cite{DezaDeza2009EncyclopediaofDistances}, which calls them \emph{$m$-hemi-metrics}---defines them in an off-by-one way from this work. In this work, the $k$ in $k$-metric refers to the arity of the function $d$, whereas in some other works it refers to the dimension $k$ of the simplex spanned by $k + 1$ affinely independent points. I.e., a $k$-metric space in this work is a $(k-1)$-metric space in~\cite{DezaRosenberg2000Semimetrics,DezaDeza2009EncyclopediaofDistances}, and in particular a (standard) metric space is a $2$-metric space in this work but a $1$-metric in theirs.}
(In fact, the definition of $k$-metrics may have been introduced even earlier, e.g., by Menger~\cite{Menger1928}.) The theory of $k$-metric spaces is rich (although somewhat disjointed). Indeed, Deza and Rosenberg~\cite{DezaRosenberg2000Semimetrics} in their work on the subject state that G\"{a}hler created an (apparently unpublished) bibliography of hundreds of works that discuss $k$-metrics. See also~\cite{DezaDeza2009EncyclopediaofDistances}.
These $k$-metric spaces are defined analogously to normal metric spaces, as the following definition makes precise.

Let $k \geq 2$ be an integer, let $X$ be a finite set, and let $d: X^k \rightarrow \R$.
We call $(X, d)$ a \EMPH{$k$-metric space}, and $d$ a \EMPH{$k$-metric function} if for any $x_1,\dots,x_k\in X$:
\begin{enumerate}[(1)]
    \item \label{item:weak-non-negative_intro} $d(x_1,\dots,x_k) \geq 0$.
    \item \label{item:weak-repeats-zero_intro} $d(x_1,\dots, x_k) = 0$ if and only if the values $x_1, \dots, x_k$ are not all distinct.
    \item \label{item:weak-permutation-invariance_intro} $d(x_1,\dots,x_k) = d(x_{\pi(1)},\dots,x_{\pi(k)})$ for any permutation $\pi : [k] \to [k]$.
    \item \label{item:weak-simplex_intro} $d(x_1,\dots,x_k) \leq  \sum_{i=1}^{k}{d(x_1,\ldots,x_{i-1},y,x_{i+1},\dots,x_k)}$ for any $y\in X$.
\end{enumerate}

It is straightforward to check that plugging $k = 2$ into the above definition yields the definition of a ``standard'' metric space, and so $k$-metrics do in fact capture and generalize metric spaces. Perhaps the most interesting aspect of this definition is its generalization of the triangle inequality $d(x_1, x_2) \leq d(x_1, y) + d(y, x_2)$ to the \EMPH{simplex inequality} $d(x_1, \ldots, x_k) \leq \sum_{i=1}^k d(x_1, \ldots, x_{i-1}, y, x_{i+1}, \ldots, x_k)$ in \cref{item:weak-simplex_intro}.
However, the condition in the simplex inequality is not the only natural generalization of the triangle inequality, and is not obviously strong enough to prove good analogs of many of the core embedding results for finite metric spaces mentioned above. 

\paragraph{Strong $k$-metric spaces.}
In this work, we introduce \emph{strong $k$-metrics}, which replace the simplex inequality (\cref{item:weak-simplex_intro}) in the definition of $k$-metric spaces with a stronger, topological condition. Furthermore, we introduce strong $k$-metric analogs of norm metrics (that is, metrics in which $d(\vec{x}, \vec{y}) = \norm{\vec{x} - \vec{y}}$ for some norm $\norm{\cdot}$) called \emph{coboundary $k$-metrics}, and of graph shortest path metrics called \emph{minimum bounding chain $k$-metrics}.
We show that these strong $k$-metric spaces ``behave nicely,'' and have many properties of regular metric spaces, which allows us to prove analogs of a number of well-known embedding results for metric spaces, including Fr\'{e}chet embedding and isometric embedding of tree metrics into $\ell_1$. See \cref{tbl:metric-class-results-summary} for a summary and \cref{sec:intro-results-techniques} for a more detailed discussion of these results. 

\paragraph{Future work and applications.} We view this paper as initial work on strong $k$-metrics, coboundary $k$-metrics, and minimum bounding chain $k$-metrics.
We hope that this work and further work in the area will result in valuable tools for solving problems in computational topology.
For example, our hope is that a strong $k$-metric analog of Bourgain's theorem~\cite{Bou85} would result in a good approximation algorithm for the topological sparsest cut problem~\cite{ParzanchevskiEtal2016IsoperSimpComp, SteenBergen2014CheegerType}.%
\footnote{There is also a more combinatorial generalization of the sparsest cut problem, for which we refer the reader to~\cite{GundertThesis2013,Gunder2015HighDimCheeger,Gundert2016Eigenvalues}.} Moreover, a variant of the Bartal tree theorem~\cite{Bartal96IntoProbTrees} could be used to solve problems about chains, such as the minimum bounding chain problem~\cite{BMN20MinHom}, via embedding into topological hypertrees.

\begin{table}[t]
\centering
\begin{tabular}{|c|l|l|}
\hline
\textbf{Symbol} & \textbf{$k$-metric space} & \textbf{Corresponding ($2$-)metric space} \\ \hline
$\W_k$ & Weak $k$-metrics & Finite metric spaces \\
$\S_k$ & Strong $k$-metrics & Finite metric spaces \\
       & Minimum bounding chain $k$-metrics & Shortest path metrics on graphs \\
$\T_k$ & Hypertree $k$-metrics & Shortest path metrics on trees \\
$\C_{k, \norm{\cdot}}$ & Coboundary $k$-metrics & Norm metric spaces (with $d(\vec{x}, \vec{y}) = \norm{\vec{x} - \vec{y}}$) \\
$\V_k$ & $(k-1)$-dimensional volume & Euclidean distance \\
\hline
\end{tabular}

\caption{Five families of (pseudo) $k$-metric spaces together with the corresponding family of metric space (i.e., $2$-metric space) that they generalize. The family of minimum bounding chain $k$-metrics is equivalent to the family of strong $k$-metrics $\S_k$ (analogously to how the family of shortest path metric spaces is equivalent to the family of finite metric spaces).} 
\label{tbl:metric-class-summary}
\end{table}

\begin{table}[t]
\centering
\begin{tabular}{|l|l|}
\hline
\textbf{The $k$-metric generalization} & \textbf{The ($2$-)metric embedding result} \\ \hline
$\C_{k, \infty} = \S_k$ (\cref{thm:frechet_embedding}) & $\ell_\infty$ metrics contain all finite metrics \\
$\C_{k, p\neq\infty}\subsetneq\S_k$ (\cref{cor:non_coboundary_metrics_exist}) & $\ell_{p\neq\infty}$ metrics do not contain all finite metrics \\
$\T_k\subseteq \C_{k, 1}$ (\cref{thm:tree_metrics_are_l1}) & Any tree metric is an $\ell_1$ metric \\  
Dimension reduction in $\C_{k, 2}$ (\cref{cor:jl-coboundary-metrics})
& Johnson-Lindenstrauss lemma \\
$\C_{k, 2}$ $(1+\eps)$-embeds into $\C_{k, p}$ (\cref{cor:Ck2-Ckp}) & $\ell_2$ $(1+\eps)$-embeds into $\ell_{p}$ \\
\hline
\end{tabular}
\caption{A list of our strong $k$-metric embedding results and the corresponding results for (standard) metrics that they generalize. The last result concerns embedding $\ell_2$ (respectively, $\C_{k, 2}$) into $\ell_p$ (respectively, $\C_{k, p}$) for $p \in [1, \infty)$ with $(1 + \eps)$ distortion.}
\label{tbl:metric-class-results-summary}
\end{table}

\subsection{Summary of Results and Techniques}
\label{sec:intro-results-techniques}
In this section, we give a summary of our results and the techniques we use to show them.  We include pointers to the main body of the paper whenever appropriate to make it easier for the reader to look into the details.
\cref{tbl:metric-class-summary} and \cref{tbl:metric-class-results-summary} give summaries of the notation that we use as well as some of the results in the paper. In this summary, we use standard terminology from algebraic topology and metric geometry. We quickly review some of this terminology, but refer the reader to \cref{sec:prelims} for detailed definitions.

\subsubsection{Strong \texorpdfstring{$k$}{k}-metrics}
\cref{item:weak-simplex_intro} in the definition of $k$-metrics is called the \EMPH{(weak) simplex inequality}.
As motivation for the definition of strong $k$-metrics, we start by noting that the triangle inequality in ``standard'' metric spaces actually enforces a stronger structural property than the weak simplex inequality does when $k > 2$ in a precise sense.
Let $(X, d)$ be a finite ($2$-)metric space, and let $G = (X, E, d)$ be the complete graph whose edges are weighted according to $d$ (i.e., with edge weights $d(u, v)$ for $(u, v) \in E$). Specifically, we note that for any $s, t \in X$, $d(s, t)$ is at most the cost of any unit $(s,t)$-flow (we prove this in \cref{lem:weak_strong_2metric_equivalency}).

A \EMPH{unit $(s, t)$-flow} $f: E \to \R$ is a function on (directed) edges of $G$ where (1) the total flow out of the source $s$ is $1$, (2) the total flow into the sink $t$ is $1$, and (3) flow is conserved at all vertices $v \in V \setminus \set{s,t}$.
The \EMPH{cost} of a (unit) $(s, t)$-flow $f$ is $\sum_{(u, v) \in E} \abs{f(u, v)} \cdot d(u, v)$.
An $(s,t)$-path is captured by the special case of a unit flow where the flow values are binary, i.e., where $f(u, v) \in \bit$ for all $u, v \in V$.
Thus, $d(s, t)$ being at most the cost of any unit $(s, t)$-\emph{flow} is a (not necessarily strictly) stronger condition than 
$d(s,t)$ being at most the length of any $(s,t)$-\emph{path} in $G$, which in turn is a (not necessarily strictly) stronger condition than the triangle inequality (which considers paths of length $2$).
We call the first of these conditions---that $d(s, t)$ is at most the cost of any unit $(s, t)$-flow---the \emph{strong} triangle inequality 
Yet, for (standard) metrics one can show that these three conditions---the strong triangle inequality, that $d(s, t)$ is at most the length of any $(s, t)$-path, and the (standard) triangle inequality---are all equivalent. 
Moreover, many metric embedding results crucially rely on this equivalence. However, unfortunately, the analogous equivalence does not hold for $k$-metrics with $k > 2$.

\begin{figure}[t]
    \centering
    \includegraphics[height=1.1in]{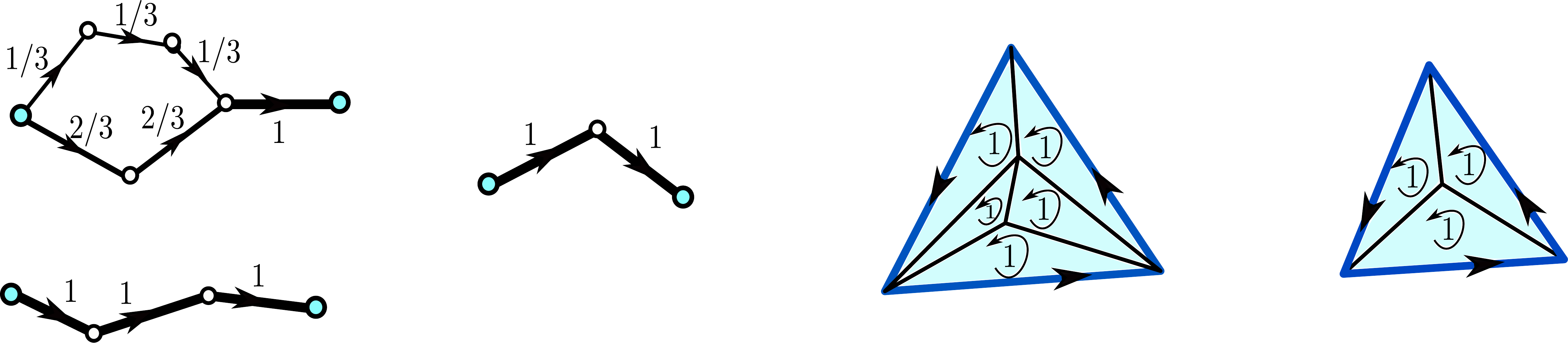}
    \caption{Left: two $s,t$-flows, the top one splits the unit flow between two paths, the bottom one is just a path. Both are checked by the strong triangle inequality but not the weak one.
    Middle-left: a $s,t$-flow that happens to be a path of length two, checked by the strong and weak triangle inequality.
    Middle-right: a $2$-chain whose boundary is the blue triangle, checked by the strong simplex inequality but not the weak one.
    Right: a simpler $2$-chain whose boundary is a the blue triangle, checked by the strong and weak simplex inequality.}
    \label{fig:chains}
\end{figure}

\paragraph{The strong simplex inequality.}
We next define a notion of $k$-metrics that \emph{does} enforce the higher-dimensional analog of the strong triangle inequality, which we call the \emph{strong simplex inequality}. To give this definition, we first need to define the higher-dimensional analog of flows. %
To do this, we use the language of algebraic topology.
For formal definitions and a more detailed summary, see \cref{sec:prelims-algebraic-topology}.

With each $k$-metric $(X, d)$, we associate a complete \EMPH{$(k-1)$-simplicial complex} $K$, i.e., the set of all subsets of $X$ with cardinality at most $k$.
Subsets of $X$ of cardinality $k$ are called the \EMPH{$(k-1)$-simplices} of $K$, e.g.~vertices are $0$-simplices and edges are $1$-simplices. We denote the set of all $(k-1)$-simplices of $K$ by $K_{k-1}$. A \EMPH{$(k-1)$-chain} $f$ is a real-valued function on the (oriented) $(k-1)$-simplices of $K$
(e.g.,~flows are $1$-chains). 
We denote the space of all $(k-1)$-chains in $K$ by $C_{k-1}[K]$.
The \EMPH{boundary} of
a $(k-1)$-simplex $t$ is a $(k-2)$-chain that 
is $\pm 1$ 
(according to the orientation of $t)$ on all $(k-2)$-faces of $t$ (all subsets of $t$ with cardinality $k-1$) and $0$ everywhere else.
We also extend the definition of boundary from individual simplices to chains by defining the boundary of a $(k-1)$-chain $\alpha$ to be the weighted sum of the boundaries of its simplices. We denote the boundary of $\alpha$ by $\partial_{k-1}\cdot\vec{\alpha}$, and remark that $\partial_{k-1}$ is a linear operator.
Finally, we define the \emph{cost} of a $(k-1)$-chain $\alpha$ to be the dot product $|\vec{\alpha}|\cdot \vec{d} = \sum_{\tau\in K_{k-1}} |\alpha[\tau]|\cdot d(\tau)$,
where $|\vec{\alpha}| := (|\alpha[\tau]|)_{\tau \in K_{k-1}}$ and $\vec{d} := (d(\tau))_{\tau \in K_{k-1}}$.
See \cref{sec:prelims} for more formal definitions.

See also \cref{fig:chains} for examples of unit $(s, t)$-flows (which are $1$-chains) and of $2$-chains whose boundaries are three-edge-cycles (i.e.,~cycles of length three). These help illustrate (chains quantified by) the strong simplex inequality when $k$ is $2$ and $3$.
We now formally present the definition of the \emph{strong simplex inequality} and \emph{strong $k$-metric spaces}.

Let $X$ be a finite set, and let $d: X^k \rightarrow \R$.
We call $(X, d)$ a \EMPH{strong $k$-metric space}, and $d$ a \EMPH{strong $k$-metric function} if any elements $x_1, \ldots, x_k \in X$ satisfy~\cref{item:weak-non-negative_intro,item:weak-repeats-zero_intro,item:weak-permutation-invariance_intro} in the definition of a $k$-metric space, and the following \EMPH{strong simplex inequality}.
\begin{enumerate}
    \item [$4'$.] \label{item:strong-simplex-intro} Let $K$ be the complete 
    $(k-1)$-dimensional simplicial complex on the vertex set $X$. Let $t$ be the $(k-1)$-simplex in $K$ with vertices $x_1, \ldots, x_k$, and let $\vec{\alpha} \in C_{k-1}(K)$ be such that $\partial_{k-1}\cdot \vec{\alpha} = \partial_{k-1}\cdot \1_t$ (i.e., the boundary of the $(k-1)$-chain $\vec{\alpha}$ is the same as the boundary of $\vec{e}_t$).
    Then
    \[
    d(x_1, \ldots, x_k) = d(t) \leq |\vec{\alpha}|\cdot \vec{d} = \sum_{\tau\in K_{k-1}} |\alpha[\tau]|\cdot d(\tau) \ \text{,}
    \]
    where 
    $|\vec{\alpha}| := (|\alpha[\tau]|)_{\tau \in K_{k-1}}$ and
    $\vec{d} := (d(\tau))_{\tau \in K_{k-1}}$.
\end{enumerate}
Following standard terminology for metrics, we say that $(X, d)$ is a (strong) \EMPH{pseudo} $k$-metric if \cref{item:weak-repeats-zero_intro} is replaced with the weaker property $d(x_1,\dots,x_k) = 0$ \emph{if} $x_i = x_j$ for any $i \neq j$, and that it is a \EMPH{meta} (strong) $k$-metric if \cref{item:weak-repeats-zero_intro} is replaced with the weaker property $d(x_1,\dots,x_k) = 0$ \emph{only if} $x_i = x_j$ for any $i \neq j$. Following standard practice, we sometimes drop the word ``pseudo'' and do not differentiate between pseudo and non-pseudo $k$-metric spaces in what follows.

We use ${\cal S}_k$ to denote the family of all strong pseudo $k$-metric spaces.
We use $\W_k$ to denote the family of all pseudo $k$-metric spaces, which we sometimes call \EMPH{weak pseudo $k$-metric spaces} for contrast with strong pseudo $k$-metric spaces.
We show that, as the names suggest, the strong simplex inequality is in fact stronger than (at least as tight as) the (weak) simplex inequality and therefore that ${\cal S}_k \subseteq {\cal W}_k$ (\cref{lem:simplex_boundary,cor:strong-k-metric-is-k-metric}). We additionally show that ${\cal S}_2 = {\cal W}_2$ coincides with the family of all finite pseudo metric spaces and therefore strong $k$-metric spaces generalize (standard) metrics spaces (\cref{cor:S2_equals_W2}). On the other hand, we show that ${\cal S}_k \subsetneq {\cal W}_k$ for $k \geq 3$ (\cref{cor:weak-k-not-strong-k}).
We also show that strong $k$-metric spaces can be verified in polynomial time, via solving multiple linear programs (\cref{lem:verify_strong_metrics}).

Similar to graph shortest paths metrics, we can define \EMPH{minimum bounding chain $k$-metrics} for a given $(k-1)$-simplicial complex with positive weights on its $(k-1)$-simplices, with complete $(k-2)$-skeleton, in which every $(k-2)$-cycle is a boundary cycle.
Let $K$ be such a complex with vertex set $X$.  
The minimum bounding chain $k$-metric $d$, assigns to a simplex $t = (x_1, \ldots, x_k)$ the minimum cost of any $(k-1)$-chain with boundary $\partial t$.  A minimum bounding chain satisfies the strong simplex inequality by its definition, as well as other properties of a strong $k$-metric. If we allow (non-negatively) weighted $(k-1)$-simplices, then the minimum bounding chain $k$-metrics can express all (finite) strong $k$-metrics.

\paragraph{Examples of (strong) $3$-metrics.}
We next give several simple examples of strong (potentially pseudo or meta) $k$-metrics. 
An important and intuitive example of a pseudo $3$-metric space $(x, d)$ is what we call an \emph{area metric}. For an area metric, the set $X$ is a finite subset of $m$-dimensional Euclidean space ($X \subset \R^m$) for some $m$, and $d(\vec{x}, \vec{y}, \vec{z})$ is defined to be the area of the triangle spanned by points $\vec{x}, \vec{y}, \vec{z} \in X$.
We also generalize this definition and define \emph{volume $k$-metrics} by defining $d(\vec{x}_1, \ldots, \vec{x}_k)$ to be the ($(k-1)$-dimensional) volume of the convex hull of $\vec{x}_1, \ldots, \vec{x}_k$, which is a $k$-simplex.
We focus on studying such $k$-metrics (i.e., $\V_k$ metrics) in~\cref{sec:volume_metric}, where in particular we show that volume $k$-metrics are strong (pseudo) $k$-metrics.

One can also check that the following other natural examples are strong \emph{meta} $3$-metric spaces:
\begin{enumerate}[(1)]
    \item \label{item:vr_perim} $X\subset \R^m$, and $d(\vec{x}, \vec{y}, \vec{z})$ is the perimeter of the triangle spanned by $\vec{x}, \vec{y}, \vec{z} \in X$.
    \item \label{item:vr_diam} $X\subset \R^m$, and $d(\vec{x}, \vec{y}, \vec{z})$ is the maximum side length of the triangle spanned by $\vec{x}, \vec{y}, \vec{z} \in X$.
    \item \label{item:cech_diam} $X\subset \R^m$, and $d(\vec{x}, \vec{y}, \vec{z})$ is the minimum diameter of a Euclidean ball containing $\vec{x}, \vec{y}, \vec{z} \in X$.
    \item \label{item:steiner} $G = (X, E)$ is a graph, and $d(x, y, z)$ is the
    weight of the minimum Steiner tree of $x, y, z \in X$.
\end{enumerate}
We note that the $2$-simplices in the Vietoris–Rips and \v{C}ech complexes of $X$ are exactly the triples of distinct points $\vec{x}, \vec{y}, \vec{z} \in X$ such that $d(\vec{x}, \vec{y}, \vec{z}) \leq c$ for some constant $c > 0$ with respect to the $3$-metric functions $d$ defined in \cref{item:vr_diam} and \cref{item:cech_diam}, respectively. These two simplicial complexes are extensively studied in topological data analysis.
We leave it as an open question whether the $k$-ary generalizations (for $k>3$) of the above examples are strong (meta) $k$-metrics, and conjecture that they are.

\subsubsection{Norms and Coboundary \texorpdfstring{$k$}{k}-metrics.}
We next introduce \emph{coboundary $k$-metric spaces}, a family of strong $k$-metric spaces that generalize (finite) metric spaces induced by norms.
(We show that coboundary $k$-metrics as defined are strong (pseudo) $k$-metrics in \cref{lem:k_coboundary_metrics}.)
That is, coboundary $k$-metrics generalize metric spaces $(X, d)$ satisfying $X \subseteq \R^m$ and $d(\vec{x}, \vec{y}) = \norm{\vec{x} - \vec{y}}$ for all $\vec{x}, \vec{y} \in X$ and some norm $\norm{\cdot}$. In particular, coboundary $k$-metric spaces generalize (finite) $\ell_p$ spaces.
We then discuss how coboundary $k$-metrics relate to other $k$-metrics (including other coboundary $k$-metrics). In particular, we show generalizations of some key embedding results for $\ell_p$ spaces to coboundary $k$-metrics.

A \EMPH{coboundary $k$-metric} $(X, d)$ of dimension $m$ with respect to a given vector norm $\norm{\cdot}$ is defined roughly as follows. (See~\cref{fig:cbd_metrics} for examples of $2$- and $3$-coboundary metrics of dimension $m = 2$ with respect to the $\ell_2$ norm, and see \cref{def:coboundary_metrics} for a formal definition.)
Let $K$ be the complete $(k-1)$-simplicial complex on vertex set $X$ (i.e., $K$ contains all simplices corresponding to sets of $k$ or fewer points in $X$), and assign fixed orientations to the $(k-2)$-simplicies in $K$.
Additionally, assign vectors in $\R^m$ to the $(k-2)$-simplices of $K$.
These vectors can be arranged as the rows of a matrix $F$; in particular, the rows of $F$ are indexed by the $(k-2)$-simplices of $K$.
We also note that the columns of $F$ are $(k-2)$-chains of $K$.
Then, apply the coboundary operator $\delta_{k-2}$ (column-wise) to $F$ to obtain $F' := \delta_{k-2} F$, where the rows of $F'$ are indexed by the $(k-1)$-simplices of $K$. 
The row of $F'$ indexed by a $(k-1)$-simplex s is a $(\pm 1)$-linear combinations of the vectors assigned to the $(k-2)$-faces of $s$.
(The columns of $F'$ are $(k-1)$-chains of $K$.)
Finally, to obtain our metric value on a $(k-1)$-simplex $t$, we compute the norm of the row of $F'$ that corresponds to $t$, i.e., set $d(t) = \norm{\vec{e}_t^T \cdot F'} =\norm{\vec{e}_t^T \cdot \delta_{k-2} \cdot F}$ (multiplying by $\1^T_t$ is to select the row that corresponds to $t$).

We remark that coboundary $2$-metrics with respect to $\norm{\cdot}$ are equivalent to (standard) metrics induced by $\norm{\cdot}$. Indeed, as shown in \cref{fig:cbd_metrics}, if one labels the vertices $\vec{v}$ of a complete graph with vectors $\vec{x}_v \in \R^m$, then the induced coboundary metric $d$ satisfies $d(u, v) = \norm{\vec{x}_u - \vec{x}_v}$.

\begin{figure}[t] \label{fig:coboundary}
    \centering
    \includegraphics[height=0.9in]{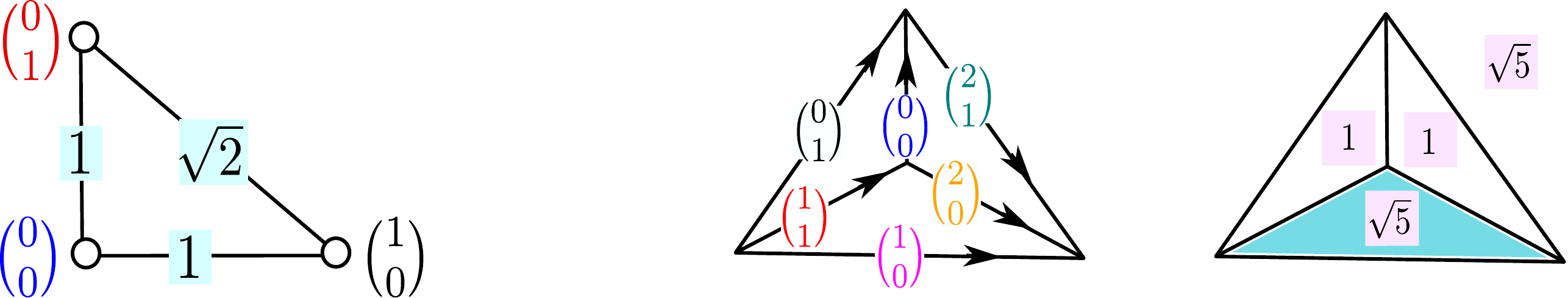}
    \caption{Left: The labeling corresponds to a pair of $0$-chains and its implied $2$-coboundary metric with the $\ell_2$ norm;
    $
    F = \begin{pmatrix} 
        \color{red}0 & \color{blue}0 & \color{black}1 \\
        \color{red}1 & \color{blue}0 & \color{black}0
    \end{pmatrix}^T
    $, 
    Right: The labeling corresponds to a pair of $1$-chains and its implied $3$-metric with $\ell_2$ norm, for example, the norm of the shaded triangle is $\Big\|\color{magenta}\begin{pmatrix}1 \\ 0\end{pmatrix}\color{black} 
    - \color{orange}\begin{pmatrix}2 \\ 0\end{pmatrix}
    - \color{red}\begin{pmatrix}1 \\ 1\end{pmatrix}\color{black}\Big\|_2 = \sqrt{5}$; 
    $
    F = \begin{pmatrix}
        0 & \color{red}1 & \color{blue}0 & \color{magenta}1 & \color{teal}2 & \color{orange}2 \\
        1 & \color{red}1 & \color{blue}0 & \color{magenta}0 & \color{teal}1 & \color{orange}0\color{black}
    \end{pmatrix}^T
    $.}
    \label{fig:cbd_metrics}
\end{figure}

We denote the space of all coboundary $k$-metrics in $m$ dimensions with respect to norm $\norm{\cdot}$ by $\C_{k, \norm{\cdot}}^m$. We slightly simplify notation for $\ell_p$ norms and write ${\cal C}_{k, p}^m$ for ${\cal C}_{k, \norm{\cdot}_p}^m$. Furthermore, we define ${\cal C}_{k, \norm{\cdot}} = \bigcup_{m\in\Z^+}{\cal C}_{k, \norm{\cdot}}^m$ and ${\cal C}_{k, p} = \bigcup_{m\in\Z^+}{\cal C}_{k, p}^m$.
As noted above, it is straightforward to show that ${\cal C}_{2, p}^m$ is the family of (finite) $\ell_p^m$ metrics in the usual sense, and that, analogously, ${\cal C}_{2, p}$ is the family of $\ell_p$ metrics.
So, coboundary $k$-metrics generalize (standard) metrics induced by norms.
There is a wealth of results regarding embeddings from and to different $\ell_p$ metrics~\cite{matousek13}. We attempt to generalize some of these results to our coboundary metrics.

\subsubsection{The Power of \texorpdfstring{$\ell_{\infty}$}{l\_infinity}}
It is well-known that any (finite) metric space is an is isometrically embedabble into $\ell_{\infty}$, via a map that is usually known as the Fr\'{e}chet embedding.
We next describe a generalization of this result that we show for coboundary metrics. Specifically, we show that any strong pseudo $k$-metric belongs to ${\cal C}_{k, \infty}$, and therefore ${\cal C}_{k, \infty} = {\cal S}_k$ (\cref{thm:frechet_embedding}).

Fr\'{e}chet's isometric embedding of finite metric spaces into $\ell_\infty$ is simple and elegant. We sketch the idea here. Let $(X, d)$ be a finite metric space, and fix $x, x'\in X$. The main idea is that there is an embedding of $X$ into the line ($\R^1$) that (1) does not expand any distance and (2) exactly preserves the distance from $x$ to $x'$. For example, the embedding that maps each $y\in X$ to $d(x,y)$ has this property. To obtain an isometric embedding of $(X, d)$ into $\ell_\infty$, one can concatenate the $\binom{n}{2}$ line embeddings corresponding to distinct pairs of elements $x, x'\in X$.%
\footnote{In fact, it suffices to use $n$ line embeddings in the Fr\'chet embedding as each line embedding preserve all distances to a single point.
But, we find this embedding (which uses $\binom{n}{2}$ line embeddings) more conducive to generalization to $k$-metrics.}
Note that each distance is never expanded and is preserved at least once, and so this does in fact give an isometric embedding into $\ell_{\infty}$.

We follow the outline of the proof above to show that any strong $k$-metric $(X, d)$ is in ${\cal C}_{k, \infty}^m$, where $m = \binom{n}{k}$ (\cref{thm:frechet_embedding}).
First, we construct analogs of the line embeddings used above. Namely, for each $k$-tuple $t = (x_1, \ldots, x_k) \in X^k$ we construct a $k$-metric $\phi_t$ in ${\cal C}_{k, \infty}^1$ that preserves the value of $d(t)$ and does not increase the value of $d(t')$ for any other such $k$-tuple $t'$ (\cref{lem:contracting_3_meteric_embedding}). 
Then, we concatenate the $m = \binom{n}{k}$ functions $\phi_t$ to obtain a $k$-metric $\phi$ in ${\cal C}_{k, \infty}^m$ that preserves all $d$ values. 

While the latter step is straightforward and similar to the case of standard metrics, the former presents a greater challenge as explicit constructions, like the mapping $y \mapsto d(x, y)$ used for standard metrics, are not readily available.

Instead, we show how to obtain such a coboundary $k$-metric as the solution of a certain (feasible) linear program. To that end, let $t = (x_1, \ldots, x_k)$. We will attempt to find a $k$-metric in ${\cal C}_{k, \infty}^1 = {\cal C}_{k, \abs{\cdot}}^1$ that preserves the value of $d(t)$ and does not expand the $d$ value on other simplices.%
\footnote{We note that all $\ell_p$ norms are equivalent in one dimension. I.e., for scalars $x \in \R$ and $p, q \in [1, \infty]$ we have that $\norm{x}_p = \norm{x}_q = \abs{x}$. Because of this, ${\cal C}_{k, \infty}^1 = {\cal C}_{k, \abs{\cdot}}^1$.}
Equivalently, we look for a $(k-2)$-chain $\vec{f}$ such that $|\delta_{k-2}\cdot\vec{f}| \leq \vec{d}$, i.e., such that $|\delta_{k-2}\cdot\vec{f}|$ assigns one non-negative value to each $(k-1)$-simplex $s$ that is not larger than $d(s)$, hence not expanding. 
(Here $|\cdot|$ and $\leq$ are treated element-wise.)
On the other hand, we want $\delta_{k-2} \cdot \vec{f}$ at $t$ to be as large as possible and ideally equal to $d(t)$.
So, we try to maximize $\1_t^T\cdot\delta_{k-2}\cdot\vec{f}$ by solving the following linear program with variables $\vec{f}$ corresponding to the $\binom{n}{k-1}$ many $(k-2)$-simplices.
\begin{align*}
\max& \quad \1_{t}^T\cdot\delta_{k-2}\cdot \vec{f} \\
\text{s.t.}&~~-\vec{d} \leq \delta_{k-2}\cdot \vec{f} \leq \vec{d} \\
&~\quad \partial_{k-2}\cdot \vec{f} = 0 \ \text{.}
\end{align*}
The equality is added for technical reasons. Namely, it ensures that the feasible region of the linear program is bounded while maintaining the same optimal value. We refer the reader to the proof of \cref{thm:frechet_embedding} for details.
Because its feasible region is a (bounded) polytope, this linear program has an optimal solution that is a vertex of this polytope.
We show that this optimal solution corresponds to a ``non-expanding'' $d$, in which $d(t)$ is preserved.

Using the fact that $\delta_{k-2} = \partial_{k-1}^T$, we rewrite the objective function of the linear program as $(\partial_{k-1}\cdot \1_t)^T\cdot \vec{f}$, and the first set of constraints as $-d(\tau)\leq(\partial_{\tau}\cdot\1_{\tau})^T\cdot\vec{f}\leq d(\tau)$ for all $(k-1)$-simplices $\tau$. 
We let $\{t_1, \ldots, t_r\}$ be the simplices whose inequality constraints are tight at the optimal solution $\vec{f}^*$, and we assume (without loss of generality) that $(\partial_{k-1}\cdot\1_{t_i})^T\cdot\vec{f} = d(t_i)$ for $i\in[r]$. Since $\vec{f}^*$ is an (optimal) solution to the linear program, the coefficient vector $\partial_{k-1} \cdot \vec{e}_t$ in the linear program must be in the cone of vectors of the tight constraint. That is, there are non-negative $\beta_1, \ldots, \beta_r$ such that,
\begin{equation} \label{eq:intro-frechet-boundary}
\partial_{k-1}\cdot \1_t
= \sum_{i = 1}^{r}{\beta_i \cdot (\partial_{k-1}\cdot \1_{t_i})}
= \partial_{k-1}\cdot\left(\sum_{i = 1}^{r}{\beta_i \1_{t_i}}\right) \ \text{.}
\end{equation}
From that, we show 
\begin{align*}
\sum_{i = 1}^{r}{\beta_i \cdot d(t_i)} &\geq d(t)
\geq (\partial_{k-1}\cdot \1_t)^T\vec{f}^*
= \sum_{i = 1}^{r}{\beta_i \cdot d(t_i)}.
\end{align*}
The first inequality uses the strong simplex inequality (the condition holds by \cref{eq:intro-frechet-boundary}), and the second holds since the coboundary metric induced by $\vec{f}^*$ is non-expanding. Therefore, we have that
$d(t)
= (\partial_{k-1}\cdot \1_t)^T\vec{f}^*$, as desired.
Again, we refer the reader to \cref{lem:contracting_3_meteric_embedding} for details.
We also remark that we crucially used the fact that $(X, d)$ was a \emph{strong} (pseudo) $k$-metric space, which we believe is good motivation for our definition of strong $k$-metrics.

\subsubsection{Other \texorpdfstring{$\ell_p$}{ell\_p} Metrics Are Not as Powerful}
Other $\ell_p$ metrics are not as expressive as $\ell_\infty$. Indeed, for any $p \neq \infty$, there is a finite metric space $(X, d)$ that is not an $\ell_p$ metric.
In this work, we show the analogous result for coboundary metrics. Specifically, we show that for any $p\neq \infty$ there is a strong pseudo $k$-metric that is not in ${\cal C}_{k, p}$ (\cref{cor:non_coboundary_metrics_exist}).

\paragraph{Apex extension.}
In order to prove this, we use a simple but powerful technique called \emph{apex extension} for constructing a $(k+1)$-metric space from a $k$-metric space in such a way that the new space shares certain properties of the original space.
Somewhat more specifically, given a $k$-metric $(X, d)$, apex extension builds a point set $X'$ and a function $d':(X')^{k+1}\rightarrow \R$ such that $(X,d)$ is in ${\cal C}_{k, p}$ if and only if $(X', d')$ is in ${\cal C}_{k+1, p}$. We can then combine this technique with examples of non-$\ell_p$-metric spaces (i.e., examples of finite metric spaces not isometrically embeddable into $\ell_p$) to construct strong $k$-metric spaces that do not belong to ${\cal C}_{k, p}$ for all $k \geq 3$.
See \cref{fig:apex_ext_1} for an illustration of apex extension.

Let $(X, d)$ be a $k$-metric, and let $a$ be an element not in $X$ that we call \EMPH{apex}.
Now, let $X' = X\cup\{a\}$, and let $d'$ be a function on $(k+1)$-tuples of $X'$ defined as follows. 
\begin{enumerate}[(i)]
    \item For any $x_1,\dots,x_{k+1} \in X'$, if $x_1, \ldots, x_{k+1}$ are not distinct or do not include $a$, $d'(x_1,\dots,x_{k+1}) = 0$.
    \item Otherwise, if $x_i = a$ for some $i$, $d'(x_1,\dots,x_{k+1}) = d(x_1, \ldots, x_{i-1}, x_{i+1}, \ldots, x_{k+1})$.
\end{enumerate}
We call $(X', d')$ the \EMPH{apex extension} of $(X, d)$.  In addition to the fact that $(X', d')$ is a weak $(k+1)$-metric (\cref{lem:apex_extension_weak_to_weak}), we show that
    $(X', d')$ is in ${\cal C}_{k+1, \norm{\cdot}}$ for any $k\geq 2$ and any norm if any only if $(X, d)$ is in ${\cal C}_{k, \norm{\cdot}}$ (\cref{lem:apex_extension_coboundary_to_coboundary}).
Since ${\cal S}_k = {\cal C}_{k, \infty}$, the statement above in particular implies that $(X', d')$ is a strong pseudo $(k+1)$-metric if and only if $(X, d)$ is a strong pseudo $k$-metric (\cref{cor:apex_extension_strong_iff_strong}). 
We use this fact to build pseudo $k$-metrics that are not strong pseudo $k$-metrics for $k>3$ from an explicit $3$-metric that is not a strong $3$-metric (\cref{lem:weak-strong-3-metric,cor:weak-k-not-strong-k}).

\begin{figure}[t]
    \centering
    \includegraphics[height=1.2in]{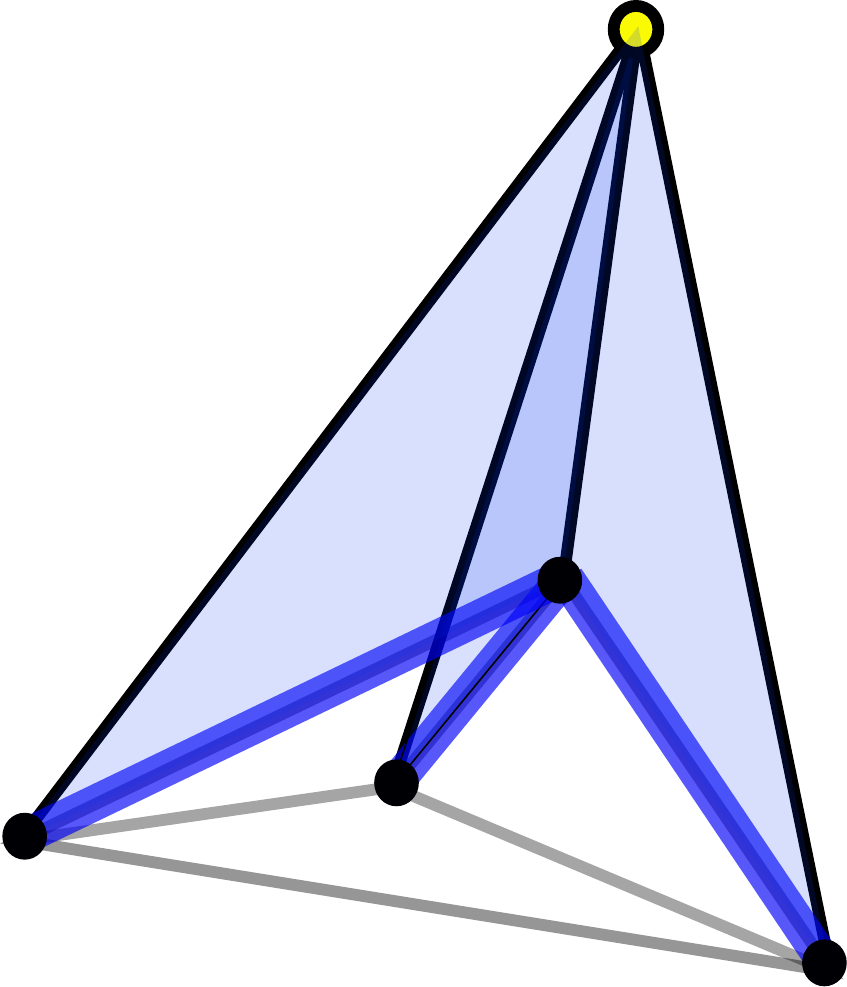}
    \caption{
        A visualization of apex extension.
        The initial metric is the coboundary (pseudo) $2$-metric induced by the black vertices, which are $0$-simplices. It has value one on the bold blue edges and zero on other edges.
        By including the apex vertex (shown at the top in yellow), this is transformed into a coboundary $3$-metric that has value one on the shaded blue triangles and zero on all other triangles. 
    }
    \label{fig:apex_ext_1}
\end{figure}

\subsubsection{Other Generalizations of Embedding Results} 
We also generalize several other well-known results from the study of metric embeddings.
First, we give a ``meta'' result for converting linear norm embedding results to results for coboundary metrics (\cref{prop:norm-embeddings-coboundary}). From this, we get two corollaries: (1) a generalization of the Johnson-Lindenstrauss lemma~\cite{JL84} for dimension reduction in $\ell_2$ and (2) a generalization of embedding $\ell_2$ into $\ell_p$ nearly isometrically, which is (a strengthening of a special case of) Dvoretzky's theorem~\cite{Dvoretzky60,FigielLM77}.
Specifically, we show any $n$-point $k$-metric in ${\cal C}_{k, 2}^m$ is $\varepsilon$-close to (1) a $k$-metric ${\cal C}_{k, 2}^{O(k\log n/\varepsilon^2)}$ (\cref{cor:jl-coboundary-metrics}), and (2) a $k$-metric in ${\cal C}_{k, p}^{m'}$ for $p \in [1, \infty)$ and $m' = \poly(m, n)$ (\cref{cor:Ck2-Ckp}).

Second, we show the generalization of the fact that all tree metrics are $\ell_1$-metrics (see, e.g.,~\cite[Chapter 1, Exercise 4]{matousek13}). 
We start by defining a higher-order analog of tree metrics called \emph{hypertree $k$-metrics}, the family of which we denote by ${\cal T}_k$. We then show that
${\cal T}_k\subseteq {\cal C}_{k, 1}$ (\cref{thm:tree_metrics_are_l1}).
Hypertree $k$-metrics are a special case of minimum bounding chain $k$-metrics where the underlying complex $K$ is a $(k-1)$-hypertree, that is, $K$ does not have any $(k-1)$-cycles and all of its $(k-2)$-cycles are boundaries of $(k-1)$-chains.%
\footnote{A $(k-1)$-cycle is a non-zero $(k-1)$-chain with no boundary. In particular, $1$-cycles are circulations in a graph, i.e., flows that have net value zero on every vertex.}
In particular, $1$-trees are ``standard'' trees from graph theory, i.e., they are acyclic and connected graphs.  

\subsubsection{Volume \texorpdfstring{$k$}{k}-metrics}
\label{sec:summary-volume}

We conclude by studying what is likely the most natural generalization of Euclidean distance to $k$ points instead of $2$: the $(k-1)$-dimensional volume of the simplex spanned by the $k$ points.
More formally, we study the spaces $(X, d)$ where $X \subsetneq \R^m$ is a finite set and $d:X^k\rightarrow\R$ is the function that assigns to each $k$-tuple $(x_1, \ldots, x_k)$ the $(k-1)$-volume of the $(k-1)$-simplex with vertices $x_1, \ldots, x_k$.
We denote the family of all such metrics by $\V_k^m$, and the family of all volume $k$-metrics (in any dimension) by $\V_k = \bigcup_{m\in \Z^+}{\V_k^m}$.  

As volume $k$-metrics, $\V_k$, and coboundary $k$-metrics with respect to the Euclidean norm, $\C_{k, 2}$ spaces, are both generalizations of Euclidean distance, it is natural to compare them with each other.  
To that end, we first show that every $\V_k$ space is representable as a ${\cal C}_{k, 2}$ space (specifically, for any $m \in \Z^+$, $\V_k^m \subseteq {\cal C}_{k, 2}^{m'}$, where $m' = \binom{m}{k-1}$), showing that coboundary $k$-metric spaces are at least as expressive as volume metrics.
On the other hand, we prove several impossibility and lower bound results, such as $\C_{3, 2} \nsubseteq \V_3$ (\cref{thm:4-point-2d-coboundary}) and $\C_{4,2} \not\subseteq \V_4$ (\cref{thm:6-point-1d-coboundary}), which imply that not every $\C_{k, 2}$ space can be represented as a $\V_k$ space.
Put together, these results show that coboundary $k$-metrics are \emph{strictly} more expressive than volume metrics.

\paragraph{Volume $k$-metrics are coboundary $k$-metrics with respect to the $\ell_2$ norm.}
We sketch the proof of ${\cal V}_k^m \subseteq {\cal C}_{k, 2}^{m'}$, $m' = \binom{m}{k-1}$ for the special case of $k = 3$ (see \cref{thm:volume_to_cbd} for the full proof).
The Cauchy-Binet theorem (\cref{thm:cauchy-binet}) implies that the area of any triangle with vertices in $\R^m$ equals the $\ell_2$ norm of the vector of the areas of its $\binom{m}{2}$ projections into axis-aligned planes (\cref{cor:cauchy-binet-volume}); see the left image in \cref{fig:area_cbd} for an illustration of these projections when $m = 3$.
This allows us to reduce showing that $\V_3^m \subseteq \C_{3, 2}^{m'}$ ($m' = \binom{m}{k-1} = \binom{m}{2}$) to showing that $\V_3^2 \subseteq \C_{3, 2}^1 = \C_{3, \abs{\cdot}}^1$.

To this end, we define the $1$-chain $f$ on the edges of the complete $2$-complex $K$ with vertex set $X$ as 
    $f(x,x')$ is the signed area of the triangle $(o, x, x')$,
where $o$ is the origin.
The signed area of a triangle is the area of the triangle times $1$ if the $(o, x, x')$ is counter clockwise and $-1$ otherwise.
See the middle-left image in \cref{fig:area_cbd} for an example.

We then show that the coboundary $3$-metric $|\delta_1\cdot\vec{f}|$ is the vector of areas of the triangles, which is what we need (here $|\cdot|$ denotes element-wise absolute value). Specifically, we show the area of any triangle $t$ equals $|\1_t^T\cdot\delta_1\cdot\vec{f}|$. See the middle, middle-right and right images of \cref{fig:area_cbd} for illustration of different cases.

\begin{figure}[t]
    \centering
    \includegraphics[height=0.9in]{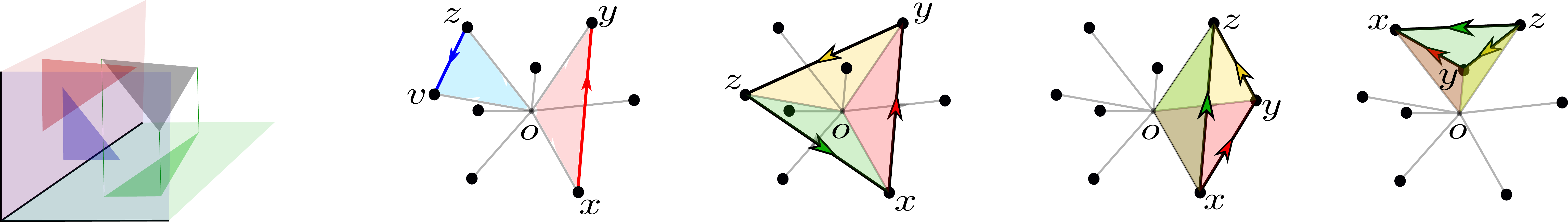}
    \caption{
        Left: a triangle an its projection into $\binom{3}{2}$ axis-aligned planes.
        Middle-left: $f(x, y)$ equals the area of the blue triangle, $f(z, v)$ is the area of the red triangle, $f(y,x) = -f(x,y)$, and $f(v, z) = -f(z, v)$.
        Middle: the coboundary $3$-metric at $(x, y, z)$ is $f(x, y) + f(y, z) + f(z, x)$, which is the area of the triangle $(x, y, z)$.
        Middle-right: the coboundary $3$-metric at $(x, y, z)$ is $f(x, y) + f(y, z) - f(z, x)$, which is the area of the triangle $(x, y, z)$.
        Right: the coboundary $3$-metric at $(x, y, z)$ is $-f(x, y) - f(y, z) + f(z, x)$, which is the area of the triangle $(x, y, z)$.
    }
    \label{fig:area_cbd}
\end{figure}

\paragraph{Relating coboundary and volume $k$-metrics.}
We also prove several formally incomparable results showing that coboundary $k$-metrics are \emph{strictly} richer than volume $k$-metrics.
First, we show that there is a $4$-point $2$-dimensional coboundary $3$-pseudometric space with respect to the Euclidean norm (i.e., $4$-point $\C_{3, 2}^2$ space) that is not a volume $3$-pseudometric space (\cref{thm:4-point-2d-coboundary}).
Second, we show that there is a $6$-point $1$-dimensional coboundary $4$-pseudometric space (i.e., $6$-point $\C_{4, \abs{\cdot}}^1$ space) that is not a volume $4$-pseudometric space (\cref{thm:6-point-1d-coboundary}).
We note that these results are incomparable; the first result holds for smaller arity $k$ and number of points $n$ (i.e., $k = 3$ and $n = 4$ versus $k = 4$ and $n = 6$), whereas the second result holds for smaller dimension $m$ (i.e., $m = 1$ versus $m = 2$).
Third, we show that there is a family of $n$-point $1$-dimensional coboundary $3$-pseudometric spaces (i.e., $\C_{3, \abs{\cdot}}^1$ spaces) that are not only not realizable as volume metric spaces using $m = o(\log n)$ dimensions, but not even embeddable into such volume metric spaces with any constant distortion (\cref{cor:one-cbd-many-vol}).
Although we typically treat $k$-pseudometric spaces the same as $k$-metric spaces, here we emphasize that the first two results hold for \emph{$k$-pseudometric} spaces, whereas the third holds for (non-pseudo) \emph{$k$-metric} spaces.

\paragraph{A $3$-metric in ${\cal C}_{3, \abs{\cdot}}^1$ that is not a low-dimensional volume metric.}
We now summarize the proof of the third result described above.
Let $(X, d)$ be the $n$-point $3$-metric that is one on all triples of distinct points in $X$, i.e., the ``discrete $3$-metric'' on $X$.
As evidence that coboundary metrics are richer than volume metrics, we show that $(X, d) \in {\cal C}_{3, 2}^1$ and yet $(X, d)$ requires $\Omega(\log n)$ dimensions to be realized as a volume $3$-metric. (In fact, we can show the stronger statement that any embedding of $(X,d)$ into the volume metric with $m \in o(\log n)$ dimensions requires super-constant distortion, although we focus on the simpler result.)

To see the former, let $K$ be a complete $2$-simplicial complex and orient edges so that there is no directed cycle in its $1$-skeleton (its $1$-skeleton is a tournament). Now, consider the $1$-chain $f$ that assigns $1$ to all these oriented edges.  Since the tournament has no directed cycle, $\delta_1 f$ is either $1$ or $-1$ on all triangles, and therefore $|\delta_1 \cdot\vec{f}|$ is a $3$-metric that is one on all distinct triples of points in $X$. 

Now, let $g:X\rightarrow \R^m$ be any isometric embedding of $(X, d)$ into a $\V_k^m$ space.
Let $x, x' \in X$ be such that $g(x), g(x')$ is a furthest pair of points in $g(X)$ (i.e., such that $\norm{g(x) - g(x')}_2$ is maximized), and let $\ell=\norm{g(x) - g(x')}_2$.
Now, consider any third point $y\in X \setminus \set{x, x'}$.  Since the area of the triangle $(g(x), g(x'), g(y))$ is one, $g(y)$ is on on the surface of the cylinder with axis $(g(x), g(x'))$ that has radius $2/\ell$.  Thus, all the points are on the surface of this cylinder. 
We apply the linear transformation that shrinks by a factor of $\ell$ in the direction of $g(x) - g(x')$ and scales by a factor of $\ell$ in every orthogonal direction.  
If we assume without loss of generality that $g(x) = \vec{0}$ and $g(x') = \ell \cdot \vec{e}_1$, then this linear transformation is expressed by the diagonal matrix $\diag(1/\ell, \ell, \ell, \cdots, \ell)$.
After the transformation, all the points are on the surface of a cylinder of constant height and constant radius.  So, they are within a ball of constant radius. Further, our transformation has the property that it does not shrink the area of any triangle.

So, we obtain $n$ points in a ball of constant radius such that the area of the triangle of each triple of them is \emph{at least} one.
These two conditions imply a constant lower bound on the distance between every pair of points. Since everything is in a ball of constant radius, a simple packing argument implies that the dimension of the space is $m \geq \Omega(\log n)$.

\subsection{Related Work} \label{sec:related-work}

As we have already discussed, prior work on $k$-metrics is closely related to this work. See, e.g.,~\cite{Gahler-63, Gahler-64, DezaRosenberg2000Semimetrics, DezaDeza2009EncyclopediaofDistances, Mustafa2006ANA}. We also again note that according to~\cite{DezaRosenberg2000Semimetrics}, around 1990 G\"{a}hler collected a bibliography of over a hundred works related to $k$-metrics, and it seems that a substantial amount of additional work has been done since then (e.g.,~\cite{Mustafa2006ANA}). We do not attempt to summarize this extensive body of work, but note that, to the best of our knowledge, none of it has defined or studied our topological (i.e., strong) $k$-metrics.

Past that, the large body of work on (finite) metric spaces and (algorithmic applications of) metric embeddings is related to this paper. We refer the interested reader to the books and surveys~\cite{matousek13, IndykMatousek17Handbook, Linial2003Survey} for summaries of this body of work. Indeed, we have provided higher-order analogs of most of the standard concepts in metric spaces, and it will be interesting to see which results on metric embeddings can be adapted to these higher-order analogs and what algorithmic applications these results have.

Additionally, we note two other notions of ``higher-order metrics'' besides $k$-metrics.
First, we note the work of Feige~\cite{Feige1998ApxBandwidth}, which defines a notion of ``volume'' for finite metric spaces in terms of embeddings into Euclidean space and studies volume-preserving embeddings. 
More specifically, Feige defines the volume of $k$ points in a finite metric space $(X, d)$ to be the maximum volume of the convex hull of their images for any non-expanding embedding $f$ into Euclidean space (i.e., embedding $f$ such that $\norm{f(x) - f(y)} \leq d(x, y)$ for all $x, y \in X$). Then, he uses embeddings that (nearly) preserve this notion of volume to design an approximation algorithm for the minimum bandwidth problem. While our work is not directly related to this type of (nearly) volume-preserving embedding, we pose as an open question whether Feige's notion of volume is a (strong) $k$-metric.

Second, we note the work of Bryant and Tupper~\cite{Bryant-Tupper-12-diversities,BryantTupper2014Diversity} on \emph{diversities}. Diversities are spaces $(X, d)$ where $d$ is a function from (finite) subsets $S \subseteq X$ to the non-negative reals that satisfies the ``triangle inequality'' condition $d(A \cup C) \leq d(A \cup B) + d(B \cup C)$ for all finite subsets $A, B, C \subseteq X$ with $B$ non-empty.
They note that the restriction of $d$ to subsets $S$ of size $2$ induces a ``standard'' metric, and so one can view diversities as a different generalization of metrics from (strong) $k$-metrics.
Indeed, we note that diversities are substantially different from (strong) $k$-metrics. In particular, the $d$ in diversities is defined on sets of elements of different sizes, and the ``triangle inequality'' in diversities upper bounds a given evaluation of $d$ as the sum of two other evaluations of $d$ (as with the ``standard'' triangle inequality) as opposed to $k$ other evaluations of $d$ in the (weak) simplex inequality.
A primary goal of Bryant and Tupper's work on diversities is to extend graph algorithms based on metric embeddings to hypergraph algorithms based on diversity embeddings. 
Their definition and techniques indeed seem well-suited to hypergraph problems, however, our definition is better suited to problems in computational topology (i.e., problems on simplicial complexes) like the topological sparsest cut problem and the minimum bounding chain problem.

\subsection{Acknowledgments}
We would like to thank Mitchell Black for helpful conversations about this work.

\section{Preliminaries}
\label{sec:prelims}
In this paper, we use ideas from algebraic topology and metric geometry that we overview here. We refer the reader to~\cite{Hatcher, StillWell93Book} for a more extensive overview of algebraic topology, and to~\cite{matousek13, DezaDeza2009EncyclopediaofDistances, IndykMatousek17Handbook, Linial2003Survey}
for a more extensive overview of metric geometry.

\subsection{Algebraic Topology}
\label{sec:prelims-algebraic-topology}

A \EMPH{simplicial complex} $K$ is a set of finite subsets of a set $X=\set{x_1,\dots, x_n}$, with the property that if a subset of $X$ is in $K$ then all its subsets are in $K$ as well. I.e., $K$ is a downward-closed set system on $X$. We call an element of $K$ a \EMPH{simplex} and an element of $K$ with cardinality $k + 1$ a \EMPH{$k$-simplex}. We call $0$-, $1$-, and $2$-simplices vertices, edges and triangles, respectively. We use $K_k$ to denote the set of all $k$ simplices of $K$, and $n_k \leq \binom{n}{k+1}$ to denote the cardinality of $K_k$.  We call $K$ a $k$-simplicial complex if it has at least one $k$-simplex and no $k'$-simplex for any $k' > k$.  The \EMPH{$k$-complete} simplicial complex on set $X$ is composed of all subsets of $X$ of cardinality at most $k$. 
The \EMPH{complete} $(k-1)$-simplicial complex on $X$ is the simplicial complex that contains all $(k-1)$-simplices and no $k$-simplex. We use $K_k$ to denote the set of all $k$-simplices in a simplicial complex $K$. Further, we denote $n_k = |K_k|$, in particular, $n_0$ is the number of vertices of a complex.

Fix a (global) ordering of $X = (x_1, x_2, \ldots, x_n)$ of the vertex set $X$ of the complex $K$.
An \EMPH{oriented $(k-1)$-simplex} of the complex $K$ is given by an (ordered) sequence $(x_{i_1}, x_{i_2}, \ldots, x_{i_k})$ of its vertices. 
Two orientations of a simplex are \emph{equivalent} if their sequences have the same permutation parity, i.e., if one can be transformed into the other with an even number of transpositions. It follows that each $(k-1)$-simplex for $k \geq 1$ has exactly two orientations. The \EMPH{standard orientation} of a simplex is the one whose vertices are ordered according to the global ordering of $X$. I.e, the standard orientation of the $k$-simplex $\set{x_{i_1}, \ldots, x_{i_{k+1}}} \subseteq X$ with $i_1 < \cdots < i_{k+1}$ is $(x_{i_1}, \ldots, x_{i_{k+1}})$. 
The sign of any oriented simplex $(x_{i_1}, \ldots, x_{i_{k+1}})$ is the sign of the permutation of its indices $i_1, \ldots, i_{k+1}$. In particular, the sign of any standard oriented simplex ix $1$.

Let $\widehat{K}_k$ be the set of all oriented simplices of $K$.
A \EMPH{$k$-chain} is a function $\alpha: \widehat{K}_k \rightarrow \R$ that assigns real values to oriented $k$-simplices of $K$ such that (i) if $t'$ and $t''$ are equivalent orientations of the same simplex $t$ then $\alpha(t') = \alpha(t'')$, and (ii) if $t'$ and $t''$ are non-equivalent orientations of the same simplex $t$ then $\alpha(t) = -\alpha(t')$.  By (i) and (ii), it suffices to specify values of $\alpha$ on the standard orientation of every simplex, or, equivalently, just non-oriented simplices. %
Hence with the abuse of notation, we write $\alpha:K_k\rightarrow \R$, as a function on the non-oriented $k$-simplices.
We use $C_k(K)$ to denote the set of all $k$-chains over the complex $K$.
We can also view $\alpha$ as a vector $\vec{\alpha} := (\alpha(\sigma))_{\sigma \in K_k}$ over the set of (standard oriented) $k$-simplices. In particular, we use $\1_t$ to denote the $k$-chain that is one on the simplex $t$ and zero everywhere else.
These vectors $\1_t$ form a basis of the $n_k$-dimensional vector space spanned by the vectors $\vec{\alpha}$ for $\vec{\alpha} \in C_k(K)$.
We use the notation $\1_t^T\cdot \vec{\alpha}$ and $\vec{\alpha}[t]$ interchangeably to refer to the value that the chain $\vec{\alpha}$ assigns to $t$. Slightly abusing notation, we use this convention even if $t$ is not a simplex with standard orientation.

The \EMPH{$k$-boundary} operator, denoted $\partial_k:C_k(K) \rightarrow C_{k-1}(K)$, is a map from $k$-chains to $(k-1)$-chains.  Since $\partial_k$ is linear, it suffices to define its action on every $k$-simplex $t = (x_1, x_2, \ldots x_k)$.
\[
\partial_k(\1_{t}) = 
\partial_k(\1_{x_1, x_2, \ldots x_k}) = 
\sum_{i=1}^{k}{(-1)^{i+1} \cdot \1_{x_1, x_2, \ldots x_{i-1}, x_{i+1}, \ldots x_k}}.
\]
The $k$-boundary operator $\partial_k$ is a linear operator that can be presented by a $n_{k-1}\times n_k$ matrix; we slightly overload notation and use $\partial_k$ to refer to this matrix as well.  The transpose of this matrix is called the \EMPH{$(k-1)$-coboundary} operator, denoted $\delta_{k-1} = \partial_k^T$. We will use the coboundary operator to define coboundary $k$-metrics.

Consider $\partial_k:C_k(K)\rightarrow C_{k-1}(K)$, $\delta_k:C_k(K)\rightarrow C_{k+1}(K)$, and their transposes $\delta_{k-1}$ and $\partial_{k+1}$, respectively.
The vector space $C_k(K)$ has the following four natural subspaces: (1) $\im(\partial_{k+1})$, the \EMPH{boundary space}, (2) $\ker(\partial_k)$, the \EMPH{cycle space}, (3) $\im(\delta_{k-1})$, the \EMPH{coboundary space}, and (4) $\ker(\delta_k)$, the \EMPH{cocycle} space.  
An important property of the boundary and coboundary operators is that $\partial_{k-1}\partial_k = 0$, and similarly $\delta_k\delta_{k-1} = 0$.  It follows that the boundary space is a subspace of the cycle space, and the coboundary space is a subspace of the cocycle space.  Further, since $\partial_k = \delta_{k-1}^T$, $C_k(K)$ orthogonally decomposes to the cycle space and the coboundary space.  Similarly, since $\delta_k = \partial_{k+1}^T$, $C_k(K)$ orthogonally decomposes into the cocycle space and the boundary space. We use $\partial_k[K]$ and $\delta_k[K]$ to emphasize that the operators are with respect to the complex $K$, when the complex is not clear from the context.

For a $k$-chain $\alpha$ and a $(k-1)$-chain $\beta$ in a complex $K$ such that $\partial_k\vec{\alpha} = \vec{\beta}$, we say that $\alpha$ is a \EMPH{bounding chain} of $\beta$. If $\vec{w}$ is a vector of non-negative weights indexed by the $(k-1)$-simplices $t$ of $K$, the \EMPH{cost} of $\alpha$ with respect to $\vec{w}$ is $\abs{\vec{\alpha}} \cdot \vec{w} = \sum_{t \in K_{k-1}} \abs{\alpha[t]} \cdot w(t)$.
A \EMPH{minimum bounding chain} of $\beta$ is a bounding chain of $\beta$ with minimum cost.

We also use special terminology for graphs that are $1$-simplicial complexes.  A $1$-chain is equivalent to a $\EMPH{flow}$.  Let $s$ and $t$ be vertices of a graph, and $f$ a $1$-chain (or flow) such that $\partial_1 f = c(\1_t - \1_s)$ for a non-negative constant $c$.  We call such an $f$ an \EMPH{$(s,t)$-flow} with \EMPH{value} $c$. We refer to it as a \EMPH{unit $(s,t)$-flow} when $c = 1$. A $1$-cycle as defined above coincides with the notion of a closed walk in graphs.  
We specifically use \EMPH{simple cycle} to refer to closed walks with no repeated vertices in a graph.

\subsection{Metric Geometry}

\begin{definition}[$k$-metric spaces]
\label{def:weak_k_metric}
Let $X$ be a finite set, and let $d: X^k \rightarrow \R$.
We call $(X, d)$ a \EMPH{$k$-metric space}, and $d$ a \EMPH{$k$-metric function} if for any $x_1,\dots,x_k\in X$:
\begin{enumerate}[(1)]
    \item \label{item:weak-non-negative} $d(x_1,\dots,x_k) \geq 0$.
    \item \label{item:weak-repeats-zero} $d(x_1,\dots, x_k) = 0$ if and only if the values $x_1, \dots, x_k$ are not all distinct.
    \item \label{item:weak-permutation-invariance} $d(x_1,\dots,x_k) = d(x_{\pi(1)},\dots,x_{\pi(k)})$ for any permutation $\pi : [k] \to [k]$.
    \item \label{item:weak-simplex} $d(x_1,\dots,x_k) \leq  \sum_{i=1}^{k}{d(x_1,\ldots,x_{i-1},y,x_{i+1},\dots,x_k)}$ for any $y\in X$.
\end{enumerate}
\end{definition}
A $k$-metric $(X, d)$ is a \EMPH{pseudo} $k$-metric if \cref{item:weak-repeats-zero} is replaced with the weaker property $d(x_1,\dots,x_k) = 0$ \emph{if} $x_i = x_j$ for any $i \neq j$, and that it is a \EMPH{meta} $k$-metric if \cref{item:weak-repeats-zero} is replaced with the weaker property $d(x_1,\dots,x_k) = 0$ \emph{only if} $x_i = x_j$ for any $i \neq j$.

We call \cref{item:weak-simplex} the \EMPH{(weak) simplex inequality} to differentiate it from the strong simplex inequality that takes its place in strong $k$-metrics.
For the same reason, sometime, we refer to $k$-metric spaces as \emph{weak} $k$-metric spaces.
The $2$-metric spaces are \EMPH{standard metric spaces} that have been the main object of study in metric geometry.

Let $(X, d)$ and $(X', d')$ be finite (pseudo) $k$-metric spaces, and let $f:X\rightarrow X'$.
We say that $f$ is an \EMPH{isometry} if for all $(x_1, \ldots, x_k) \in X^k$, $d(x_1, \ldots, x_k) = d'(f(x_1), \ldots, f(x_k))$.

A \EMPH{norm} on $\R^m$ is a function $\norm{\cdot}:\R^m\rightarrow \R$ that satisfies the following properties for all $\vec{x}, \vec{y} \in \R^m$ and all constants $c > 0$:
\begin{enumerate}[(1)]
    \item \label{item:norm-nonneg} $\|\vec{x}\| \geq 0$, with equality if and only if $\vec{x} = \vec{0}$.
    \item \label{item:norm-multiplicative} $\|c \vec{x}\| = |c|\cdot\|\vec{x}\|$, where $|\cdot|$ is the absolute value.
    \item \label{item:norm-triangle} $\|\vec{x}\|+\|\vec{y}\|\geq \|\vec{x} + \vec{y}\|$.
\end{enumerate}
Let $X$ be a finite subset of $\R^m$ and $\norm{\cdot}$ any norm on $\R^m$. 
The norm induces a metric $d_{\norm{\cdot}}$ on all points of $\R^m$, and in particular all points of $X$, where for any $x, x'\in \R^m$, $d_{\norm{\cdot}}(x, x') = \norm{x - x'}$. In this case, $(\R^m, d_{\norm{\cdot}})$ is a metrics space and $(X, d_{\norm{\cdot}})$ is a finite metric space.

For any real $p \geq 1$, the \EMPH{$\ell_p$ norm} of $\vec{x} \in \R^m$ is defined as 
\[
\norm{\vec{x}}_p := \Big(\sum_{i=1}^{m}{\abs{x_i}^p\Big)^{1/p}} \ \text{,}
\]
and the \EMPH{${\ell_\infty}$ norm} of $\vec{x}$ is defined as $\norm{\vec{x}}_\infty: =\max_i(\abs{x_i})$.
Let $d_p = d_{\norm{\cdot}_p}$ be the distance function implied by the $p$-norm.
An $\ell_p^m$ metric is a finite metric space that admits an isometric embedding into $(\R^m, \norm{\cdot}_p)$ and an \EMPH{$\ell_p$ metric} is a finite metric space that admits an isometric map into $(\R^m, \norm{\cdot}_p)$ for some $m\in\Z^+$. Euclidean metrics are $\ell_p$ metrics with $p = 2$.
\section{Strong \texorpdfstring{$k$}{k}-metric Spaces}
We start by calling recalling the definition of a strong $k$-metric space.
\begin{definition}[Strong $k$-metric spaces] \label{def:strong-k-metric}
Let $X$ be a finite set, and let $d: X^k \rightarrow \R$.
We call $(X, d)$ a \EMPH{strong $k$-metric space}, and $d$ a \EMPH{strong $k$-metric function} if for any $x_1,\dots,x_k\in X$:
\begin{enumerate}[(1)]
    \item \label{item:strong-non-negative} $d(x_1,\dots,x_k) \geq 0$.
    \item \label{item:strong-repeats-zero} $d(x_1,\dots, x_k) = 0$ if and only if the values $x_1, \dots, x_k$ are not all distinct.
    \item \label{item:strong-permutation-invariance} $d(x_1,\dots,x_k) = d(x_{\pi(1)},\dots,x_{\pi(k)})$, for any permutation $\pi : [k] \to [k]$.
    \item \label{item:strong-simplex} 
    Let $K$ be the complete 
    $(k-1)$-dimensional simplicial complex on the vertex set $X$. Let $t$ be the (oriented) $(k-1)$-simplex in $K$ with vertices $x_1, \ldots, x_k$, and let $\alpha \in C_{k-1}(K)$ be such that $\partial\cdot \alpha = \partial\cdot \1_t$ (i.e., the boundary of the $(k-1)$-chain $\alpha$ is the same as the boundary of $t$).
    Then
    \[
    d(x_1, \ldots, x_k) = d(t) \leq |\alpha|\cdot \vec{d} = \sum_{\tau\in K_k} |\alpha[\tau]|\cdot d(\tau) \ \text{,}
    \]
    where 
    $|\alpha| = (|\alpha[\tau]|)_{\tau \in K_k}$ and
    $\vec{d} = (d(\tau))_{\tau \in K_k}$.
\end{enumerate}
\end{definition}

\subsection{Relating Weak and Strong \texorpdfstring{$k$}{k}-metrics}
First, we show that strong $k$-metrics are a subset of weak $k$-metrics, as suggested by their name. This inclusion follows from the following lemma.
\begin{lemma}
\label{lem:simplex_boundary}
Consider a $k$-simplex with vertex set $\{x_1, \ldots, x_k, y\}$. Then
\[
    \partial_{k-1}\cdot \1_{x_1, \ldots, x_k} = \partial_{k-1}\cdot\left(
        \sum_{i=1}^{k}{\1_{x_1,\ldots, x_{i-1}, y, x_{i+1},\ldots, x_k}}
    \right) \ \text{.}
\]
\end{lemma}
\begin{proof}
We have
\begin{align*}
\partial_k(\1_{x_1, \ldots, x_k, y}) &= \sum_{i=1}^{k}{(-1)^{i-1}\cdot\1_{x_1,\ldots, x_{i-1},x_{i+1}, \ldots, x_k, y}} + (-1)^{k} \cdot \1_{x_1,\ldots, x_k}
\\ &= (-1)^{k - 1} \cdot\sum_{i=1}^{k}{\1_{x_1,\ldots, x_{i-1},y,x_{i+1}, \ldots, x_k}} + (-1)^{k} \cdot \1_{x_1,\ldots, x_k}
\\ &= (-1)^{k} \cdot\left(\1_{x_1,\ldots, x_k} - \sum_{i=1}^{k}{\1_{x_1,\ldots, x_{i-1},y,x_{i+1}, \ldots, x_k}}
\right) \ \text{.}
\end{align*}
Therefore, $\1_{x_1,\ldots, x_k} - \sum_{i=1}^{k}{\1_{x_1,\ldots, x_{i-1},y,x_{i+1}, \ldots, x_k}}$ is in the image of $\partial_k$, and hence is in the kernel of $\partial_{k-1}$. The lemma follows by using the fact that $\partial_{k-1}$ is a linear operator and rearranging.
\end{proof}

\begin{corollary} \label{cor:strong-k-metric-is-k-metric}
    For every integer $k \geq 2$, any strong pseudo $k$-metric is a weak pseudo $k$-metric, i.e.,
    ${\cal S}_k\subseteq {\cal W}_k$.
\end{corollary}
\begin{proof}
Since all other properties are identical, it suffices to show that for a strong $k$-metric space $(X,d)$, if $d$ satisfies the strong simplex inequality then it also satisfies the weak simplex inequality.
To that end, let $x_1,\ldots,x_n,y\in X$, and let 
\[
    \vec{\alpha} = \sum_{i=1}^{k}{\1_{x_1,\ldots, x_{i-1}, y, x_{i+1},\ldots, x_k}} \ \text{.}
\]
By \cref{lem:simplex_boundary}, $\partial_{k-1}\cdot\vec{\alpha} = \partial_{k-1}\cdot\1_{x_1, \ldots, x_k}$.
Thus, by the strong simplex inequality,
\[
d(x_1, \ldots, x_k) \leq \sum_{\tau\in K_k} |\vec{\alpha}[\tau]|\cdot d(\tau) = \sum_{i=1}^{k}{d(x_1,\ldots, x_{i-1}, y, x_{i+1},\ldots, x_k)} \ \text{,}
\]
as needed.
\end{proof}

Next, we show a weak $3$-metric that is not a strong $3$-metric.  Hence, strong $3$-metric is strictly stronger than weak $3$-metrics, unlike strong and weak $2$-metrics which are equivalent (\cref{lem:weak_strong_2metric_equivalency}).

\begin{lemma} \label{lem:weak-strong-3-metric}
There exists a weak $3$-metric space that is not a strong $3$-metric space, i.e., ${\cal W}_3\not\subset{\cal S}_3$.
\end{lemma}

\begin{proof}
We will consider a $3$-metric space $(X = \{x_0, x_1, x_2, x_3, x_4, x_5\}, d)$, with $d$ defined below corresponding to~\cref{fig:weak-strong-3-metric}.

\begin{figure}[t]
    \centering
    \includegraphics[height=1.75in]{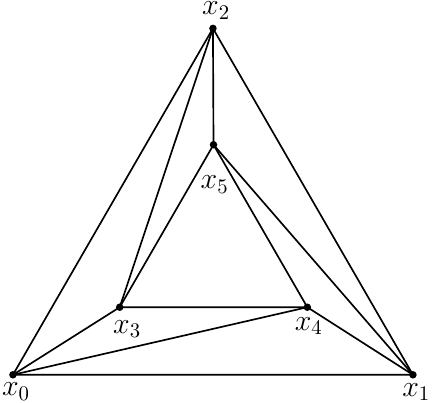}
    \caption{A subdivision of the triangle $\{x_0, x_1, x_2\}$ into seven sub-triangles. Let $d : X \times X \times X \to \R^{\geq 0}$ be such that $d(t) = 1$ for each triangle $t$ in the subdivision, and $d(t) = 10$ for each triangle $t$ (with distinct vertices) not in the subdivision (including $t = \{x_0, x_1, x_2\}$). Then $(X = \{x_0, x_1, x_2, x_3, x_4, x_5\}, d)$ gives an example of a weak $3$-metric space that is not a strong $3$-metric space.}
    \label{fig:weak-strong-3-metric}
\end{figure}

Let
\[
S := \set{\set{x_0, x_1, x_4}, \set{x_0, x_3, x_4}, \set{x_1, x_2, x_5}, \set{x_1, x_4, x_5}, \set{x_0, x_2, x_3}, \set{x_2, x_3, x_5}, \set{x_3, x_4, x_5}}
\]
denote the set of seven triangles in the subdivision of the triangle $\set{x_0, x_1, x_2}$ shown in~\cref{fig:weak-strong-3-metric}, and define
\[
d(x, y, z) :=
\begin{cases}
0 & \text{if $x = y$ or $x = z$ or $y = z$ ,} \\
1 & \text{if $\set{x, y, z} \in S$ ,} \\
10 & \text{otherwise .}
\end{cases}
\]

Now, let $\alpha$ be a $2$-chain that is $1$ on all triangles of $S$ with counterclockwise orientation,
\[
\alpha = \1_{x_0 x_1 x_4} + \1_{x_0 x_4 x_3} + \1_{x_1 x_2 x_5} + \1_{x_1 x_5 x_4} + \1_{x_0 x_3 x_2} + \1_{x_2 x_3 x_5} + \1_{x_3 x_4 x_5}.
\]
We have that $\partial\cdot \alpha = \1_{x_0 x_1} + \1_{x_1 x_2} + \1_{x_2 x_0} = \partial \cdot \1_{x_0 x_1 x_2} = \partial t$, but $d_X(t) = 10 > 7 = |\alpha|\cdot d_X$.  So, $(X, d_X)$ is \emph{not} a strong $k$-metric space, as it violates condition (3) of strong metric spaces.

On the other hand, one can check that $d(x, y, z) \leq d(w, x, y) + d(w, x, z) + d(w, y, z)$ for any $w, x, y, z \in X$, and so $(X, d)$ satisfies the (weak) simplex inequality.
The cases when $d(x, y, z) = 0$ and $d(x, y, z) = 1$ are immediate, as the sum of any three simplices is at least $1$. The case where $d(x, y, z) = 10$ follows by noting that, because at most two of the four triangles spanned by any given four points $w, x, y, z$ are in the subdivision, (at least) one term on the right-hand side must be equal to $10$.
\end{proof}

We remark that the stronger statement that $\W_k\not\subseteq \S_k$ for \emph{any} $k\geq 3$ also holds (\cref{cor:weak-k-not-strong-k}).  Proving this stronger statement, however, requires additional tools, and thus we defer it to the next section.

We also remark that, more generally, any triangulation of the triangle $(x_0, x_1, x_2)$ that does not contain $K_4$ as a subgraph can be turned into an example to show that the strong simplex inequality is strictly stronger than the weak one. Consider such a triangulation with $n$ triangles. We assign weight one to all the triangles of such a triangulation and weight $n+1$ to all other triangles of the complete $2$-simplicial complex.  Since the triangulation does not contain $K_4$ as a subgraph, the induced $2$-simplicial complex of any four vertices has at least two weight-$(n+1)$ triangles as faces, and thus satisfies the weak simplex inequality. On the other hand, the chain of all $n$ weight-$1$ triangles has boundary $(x_0, x_1, x_2)$ (since those triangles form a subdivision of $(x_0, x_1, x_2)$) and has weight $n$. This is strictly less than the weight of $(x_0, x_1, x_2)$, which is $n+1$.

Finally, we show that in contrast to $k$-metrics for $k > 2$, weak and strong $2$-metrics are equivalent. The proof follows from the folklore flow decomposition theorem.

\begin{lemma}
\label{lem:weak_strong_2metric_equivalency}
Any weak pseudo $2$-metric is a strong pseudo $2$-metric, i.e., ${\cal W}_2 \subseteq {\cal S}_2$.
\end{lemma}
\begin{proof}
The definitions of weak and strong 2-metrics are identical, except for the weak simplex inequality (\cref{def:weak_k_metric},~\cref{item:weak-simplex}) and its strong metric equivalent (\cref{def:strong-k-metric},~\cref{item:strong-simplex}). Thus, it remains to show that the former implies the latter for $2$-metrics.

Let $(X,d)$ be a $2$-metric space, and let $G$ be the weighted complete graph on vertex set $X$ where the weight of edge $(x, y)$ is $d(x,y)$. ($G$ is the graph representation of $(X, d)$.)
Let $s,t\in X$ and let $\vec{\alpha}$ be a $1$-chain on $G$ such that $\partial_1 \vec{\alpha} =\partial_1 \1_{s,t}$, i.e., such that ~$\vec{\alpha}$ is a unit $(s,t)$-flow.
We show that $|\alpha|\cdot \vec{d} = \sum_{\tau\in E} |\alpha[\tau]|\cdot d(\tau) \geq d(s,t)$, and therefore the strong simplex inequality holds.

By the flow decomposition theorem (see, e.g.,~\cite[Theorem 3.5]{Ahuja1993Book}) we can decompose $\vec{\alpha}$ into a set of cyclic flows ${\cal C}$ and a set of flows ${\cal P}$ where $P \in \mathcal{P}$ has edges $E_p$ and value $0 < f_P \leq 1$.
So, we get that 
\begin{align*}
    |\alpha|\cdot d = \sum_{\tau\in E} |\alpha[\tau]|\cdot d(\tau)&=\sum_{C\in {\cal C}} \sum_{\tau\in E}|C[\tau]|\cdot d(\tau) + \sum_{P\in {\cal P}} \sum_{\tau\in E}P[\tau]  \cdot d(\tau) \\
    &\geq \sum_{P\in {\cal P}} \sum_{\tau\in E} P[\tau] \cdot d(\tau) \\ 
    &= \sum_{P\in {\cal P}} \sum_{\tau\in E_P} f_P \cdot d(\tau) \\ 
    &= \Big(\sum_{\tau\in E_P} d(\tau) \Big) \cdot \sum_{P\in {\cal P}} f_P \\ 
    &\geq d(s,t) \sum_{P\in {\cal P}} f_P \\ 
    &=  d(s,t) & \ \text{.}
\end{align*}
The second inequality holds because the edges $\tau \in E_P$ form an $(s, t)$-path, which one can show has weight at least $d(s, t)$ by an inductive application of the triangle inequality.
The last equality holds because $\sum_{P\in {\cal P}} f_P$ is equal to the flow value of $\vec{\alpha}$ (from $s$ to $t$), which is $1$.
\end{proof}

We end this section by summarizing the inclusion and exclusion result in the following corollaries.  The first corollary follows from \cref{cor:strong-k-metric-is-k-metric} and \cref{lem:weak_strong_2metric_equivalency}.

\begin{corollary}
\label{cor:S2_equals_W2}
Any weak pseudo $2$-metric is a strong pseudo $2$-metric and vice versa, i.e., ${\cal W}_2 = {\cal S}_2$.    
\end{corollary}

\subsection{Verifying Strong \texorpdfstring{$k$}{k}-Metrics in Polynomial Time}
\label{sec:verifying-strong-poly-time}

Probably the most basic algorithmic question about finite $k$-metric spaces is whether we can verify them efficiently. More specifically, we would like to solve the following problem(s) for fixed $k \geq 2$.
Given a finite set $X$ and a function $d : X^k\rightarrow \R$ given by its evaluation table as input, decide whether $(X, d)$ is a weak or strong $k$-metric space.  

Deciding whether $(X, d)$ is a weak $k$-metric space in $O(n^{k+1})$ time is straightforward, as each condition can be checked in time only depending on $k$ for each of the $\binom{n}{k+1}$ sets of values $\set{x_1, \ldots, x_k, y} \subseteq X$, and $k$ is a constant.
Deciding whether $(X, d)$ is a strong $k$-metric is more challenging because of the strong simplex inequality. In particular, it is not clear that we can enumerate all chains $\vec{\alpha}$ with the required properties.
So, instead of enumerating all such chains $\alpha$, for every $(k-1)$-simplex $t = (x_1, \ldots, x_k)$ we compute the minimum cost of a $(k-1)$-chain that has the same boundary as $t$ and compare its cost with $d(t)$.
To do this, we use known algorithms for computing minimum bounding chains over the reals~\cite{Dey2010OHC, TahbazJadbabai2010DistCoverage} using linear programming.%
\footnote{Computing minimum bounding chains over the integers or finite fields is known to be $\NP$-hard. See, e.g.,~\cite{BMN20MinHom}.}
More specifically, we use these algorithms to compute minimum bounding chains for the boundary of every $(k-1)$-simplex $t = (x_1, \ldots, x_k) \in K_{k-1}$.
We present a detailed proof for the sake of completeness, but again note that the linear programs in the proof are already described in previous literature~\cite{Dey2010OHC, TahbazJadbabai2010DistCoverage}.

\begin{lemma} 
\label{lem:verify_strong_metrics}
Let $(X, d)$, $d: X^k \rightarrow \R$, be a finite $k$-metric space for constant $k$, with $n = |X|$.  We can check whether $d$ is a strong $k$-metric in polynomial time (by solving $O(n^k)$ linear programs, each with $n^{O(k)}$ variables and constraints).
\end{lemma}
\begin{proof}
It suffices to check, for every $(k-1)$-simplex $t$, that the optimal value of the following optimization problem is at least $d(t)$ (in fact, the optimal value would be equal to $d(t)$, realized by $\vec{\alpha} = \1_t$ among other possible optimum solutions).
\begin{equation}
\label{eqn:check_strong_metric_opt1}
\min_{\vec{\alpha}}{\sum_{\tau\in X^k}{|\vec{\alpha}[\tau]|\cdot d(\tau)}} \quad\text{s.t.}\quad \partial_{k-1}\cdot\vec{\alpha} = \partial_{k-1} \cdot \1_t.
\end{equation}
We reduce this problem to a linear programming.
First, we modify \cref{eqn:check_strong_metric_opt1} by introducing $(k-1)$-chains $\alpha^+$ and $\alpha^-$ and constraints
\[
\vec{\alpha}[\tau] = \vec{\alpha}^+[\tau] - \vec{\alpha}^-[\tau] \quad\text{and}\quad \vec{\alpha}^+[\tau], \vec{\alpha}^-[\tau] \geq 0,
\]
for every $(k-1)$-simplex $\tau$. We obtain the following modified optimization problem.
\begin{equation}
\label{eqn:check_strong_metric_opt2}
\min_{\vec{\alpha}^+, \vec{\alpha}^-}{\sum_{\tau\in X^k}{|\vec{\alpha}^+[\tau] - \vec{\alpha}^-[\tau]|\cdot d(\tau)}} \quad\text{s.t.}\quad \partial_{k-1}\cdot(\vec{\alpha}^+ - \vec{\alpha}^-) = \partial_{k-1} \cdot \1_t,
\quad\text{and}\quad \vec{\alpha}^+, \vec{\alpha}^- \geq \vec{0},
\end{equation}
where by $\vec{\alpha}^+, \vec{\alpha}^- \geq \vec{0}$, we mean that they are vectors with non-negative entries.

For any feasible solution $\vec{\alpha}^+, \vec{\alpha}^-$ of \cref{eqn:check_strong_metric_opt2} there is a feasible solution $\vec{\alpha} = \vec{\alpha}^+ - \vec{\alpha}^-$ in \cref{eqn:check_strong_metric_opt1} with the same optimization value.  On the other hand, for every feasible solution $\vec{\alpha}$ of \cref{eqn:check_strong_metric_opt1}, consider the solution $\vec{\alpha}^+, \vec{\alpha}^-$ in \cref{eqn:check_strong_metric_opt2}, where for any $\tau\in X^k$, if $\vec{\alpha}[\tau]\geq 0$ then $\vec{\alpha}^+[\tau] = \vec{\alpha}[\tau]$ and $\vec{\alpha}^-[\tau] = 0$, and otherwise $\vec{\alpha}^+[\tau] = 0$ and $\vec{\alpha}^-[\tau] = -\vec{\alpha}[\tau]$.  The optimization value of $\vec{\alpha}^+, \vec{\alpha}^-$ in \cref{eqn:check_strong_metric_opt2} is the same as the optimization value of $\vec{\alpha}$ in \cref{eqn:check_strong_metric_opt1}.  Hence, \cref{eqn:check_strong_metric_opt1} and \cref{eqn:check_strong_metric_opt2} have the same optimal value. 

Next, we show that the following LP has the same optimal value as \cref{eqn:check_strong_metric_opt2}.
\begin{equation}
\label{eqn:check_strong_metric_lp}
\min_{\vec{\alpha}^+, \vec{\alpha}^-}{\sum_{\tau\in X^k}{(\vec{\alpha}^+[\tau] + \vec{\alpha}^-[\tau])\cdot d(\tau)}} \quad\text{s.t.}\quad \partial(\vec{\alpha}^+ - \vec{\alpha}^-) = \partial t,
\quad\text{and}\quad \vec{\alpha}^+, \vec{\alpha}^- \geq 0.
\end{equation}
To that end, we observe that in both \cref{eqn:check_strong_metric_opt2} and \cref{eqn:check_strong_metric_lp} there are optimal solutions for both with the property that
    for any $\tau\in X^k$, we have $\vec{\alpha}^+[\tau] = 0$ or $\vec{\alpha}^-[\tau] = 0$. 
Let $\vec{\alpha}^+, \vec{\alpha}^-$ be any feasible solution., and let $\tau$ be any $(k-1)$-simplex.  
Let $h = \min(\vec{\alpha}^-[\tau], \vec{\alpha}^+[\tau])$.  
Replacing $\vec{\alpha}^-[\tau], \vec{\alpha}^+[\tau]$ with $\vec{\alpha}^-[\tau] - h, \vec{\alpha}^+[\tau] - h$ does not violate the constraints of the optimization problems, and it can only reduce their objective functions as $d(\tau) \geq 0$. So, it is safe to assume that for any $\tau\in X^k$, $\vec{\alpha}[\tau]^+ = 0$ or $\vec{\alpha}[\tau]^- = 0$. This implies that if $\vec{\alpha}^-[\tau] = 0$,
\[
|\vec{\alpha}^+[\tau] - \vec{\alpha}^-[\tau]| = |\vec{\alpha}^+[\tau]| = \vec{\alpha}^+[\tau] = \vec{\alpha}^+[\tau]+\vec{\alpha}^-[\tau] \ \text{,}
\]
and similarly if $\vec{\alpha}^+[\tau] = 0$,
\[
|\vec{\alpha}^+[\tau] - \vec{\alpha}^-[\tau]| = |-\vec{\alpha}^-[\tau]| = \vec{\alpha}^-[\tau] = \vec{\alpha}^+[\tau]+\vec{\alpha}^-[\tau] \ \text{.}
\]
In either case, we have 
$
|\vec{\alpha}^+[\tau] - \vec{\alpha}^-[\tau]| = \vec{\alpha}^+[\tau]+\vec{\alpha}^-[\tau]
$.

Overall, the optimization problems in \cref{eqn:check_strong_metric_opt2} and \cref{eqn:check_strong_metric_lp} have the following properties:
\begin{enumerate}[(i)]
    \item they have the same set of constraints, 
    \item they both have optimal solutions with the property above, and 
    \item for any optimal solution with this property in one of these optimization problems there is a feasible solution in the other one with the same objective functions.
\end{enumerate} 
We conclude that these two optimization problems have the same optimal value.  Moreover, \cref{eqn:check_strong_metric_lp} is a linear program with $n^{O(k)}$ variables and constraints, as needed.
\end{proof}
\section{Coboundary Metrics}
\label{sec:coboundary-metrics}

Normed vector spaces have been extensively studied as special cases of metric spaces. Such spaces have nice properties as their structures match our geometric intuition, and accordingly they appear in many applications. 
In this section, we extend the notion of normed vector spaces to $k$-metrics (\cref{sec:cbd_definition}).  We call these spaces coboundary $k$-metric spaces because of the use of coboundary operator in their definition. Note that our coboundary $k$-metrics are in fact pseudo $k$-metrics, but we refer to them as metrics for simplicity.  We denote a coboundary $k$-metric that is realized from vectors in $\R^m$ with respect to the norm $\norm{\cdot}$ by ${\cal C}_{k,\norm{\cdot}}^m$, or ${\cal C}_{k,p}^m$ when the norm is the $\ell_p$ norm for some $p \in [1, \infty]$; see \cref{def:coboundary_metrics} below for the formal definition of coboundary metrics.  
Coboundary $k$-metrics are generalizations of $\ell_p$ metrics, specifically, $\ell_p^m = {\cal C}_{2, p}^m$.
We use ${\cal C}_{k, p}$ to denote the family of all finite coboundary $k$-metrics in $m$ dimensions with respect to the $\ell_p$ norm; in particular, ${\cal C}_{2, p}$ is the space of all $\ell_p$ metrics.
The $\ell_p$ norms have been very well studied, with many insightful results that relates them via low distortion embeddings.  We describe generalizations of some significant results of this type in this section. 

In \cref{sec:cbd_subset_k_metrics}, we show that for any $p\neq \infty$, the space of ${\cal C}_{k, p}$ is strictly contained in the space of strong pseudo $k$-metrics. In \cref{sec:frechet_embedding}, we show that ${\cal C}_{k, \infty}$ contains all strong pseudo $k$-metrics spaces, hence ${\cal C}_{k, \infty} = {\cal S}_k$. This is a generalization of the Fr\'echet embedding for finite metrics, obtained when $k=2$.
In \cref{sec:norm-embeddings-to-coboundary}, we give a recipe for adapting results about norm embeddings to analogous results about embedding coboundary $k$-metrics. From this, we get generalizations of the Johnson-Lindenstrauss lemma and Dvoretzky's theorem as corollaries.
Finally, in \cref{sec:tree_metrics} we show a generalization of the fact that any tree metric is an $\ell_1$ metric.

\subsection{Definition and Basic Properties}
\label{sec:cbd_definition}
We begin by defining coboundary $k$-metrics and proving that they are strong pseudo $k$-metrics.  In particular, coboundary $2$-metrics are equivalent to normed spaces as we know them from the study of metric spaces. 
Intuitively, a coboundary $k$-metric can be constructed on a point set $X$ as follows.  We consider the complete $(k-1)$-simplicial complex on $X$.  Then we assign vectors to the $(k-2)$-simplices of $K$.  We apply the coboundary operator to these vectors to obtain vectors on $(k-1)$-simplices.  The norm of these vectors are our metric values.  \cref{fig:cbd_metrics} illustrates $2$- and $3$-coboundary metrics. 

Now, we formally define coboundary $k$-metrics. In our definition, the vectors assigned to the $(k-2)$-simplices are the rows of the matrix $F$, which is 
equivalently defined by its columns that are $(k-2)$-chains of $K$. See the example in \cref{fig:cbd_metrics}.
\begin{definition}[Coboundary $k$-metrics]
\label{def:coboundary_metrics}
Let $k \geq 2$ be an integer, let $X$ be a finite set, let $d:X^k \rightarrow \R$, and let $K$ be the complete $(k-1)$-simplicial complex with vertex set $X$.
Also, let $m\in \Z^+$ and let $\norm{\cdot}$ be any norm on $\R^m$.
We say that $(X, d)$ is a \EMPH{coboundary $k$-metric} in $m$ dimensions with respect to norm $\norm{\cdot}$, if there exist $(k-2)$-chains $f_i:K_{k-2}\rightarrow \R$ for $i\in[m]$ in $K$ such that for any distinct $x_1,\dots, x_k\in X$, 
\begin{equation}
\label{eqn:cbd_def}
    d(x_1,x_2,\ldots, x_k) = 
    \begin{cases}
    0,& \text{if there exist $i \neq j$ with $x_i = x_j$,}\\
    \norm{
        \1_{x_1,x_2,\ldots, x_k}^T\cdot \delta_{k-2} \cdot F
    },& \text{otherwise},
    \end{cases}
\end{equation}
where $F = (\vec{f}_1,\ldots,\vec{f}_m)$ is an $n_{k-2}\times m$ matrix specified by its columns, which are $(k-2)$-chains in $K$.  We say that $F$ \emph{induces} the coboundary $k$-metric $(X, d)$. For any non-distinct $x_1,\dots, x_k\in X$, $d(x_1,\dots, x_k) = 0$.  
\end{definition}

We denote the space of all coboundary $k$-metrics in $m$ dimensions with respect to the norm $\norm{\cdot}$ by ${\cal C}_{k, \norm{\cdot}}^m$.  We slightly simplify notation for the $p$-norms, denoting ${\cal C}_{k, p}^m = {\cal C}_{k, \norm{\cdot}_p}^m$. Further, we denote ${\cal C}_{k, p} = \bigcup_{m\in\Z^+}{\cal C}_{k, p}^m$.
In particular, we remark that ${\cal C}_{2, p}^m$ is the space of $\ell_p^m$-metrics, and ${\cal C}_{2, p}$ is the space of $\ell_p$-metrics, denoted $\mathbb{L}_p$ in some texts~\cite{matousek13}.
We show that coboundary $k$-metrics are strong pseudo $k$-metrics. 

\begin{lemma}
\label{lem:k_coboundary_metrics}
Let $k\geq 2$ and $m \geq 1$ be integers, and let $\norm{\cdot}$ be any vector norm. If $(X, d)$ belongs to ${\cal C}_{k, \norm{\cdot}}^m$ then $(X, d)$ is a strong pseudo $k$-metric space, i.e.,~${\cal C}_{k, p} \subseteq {\cal S}_k$.
\end{lemma}
\begin{proof}
Let $K$ be the complete simplicial complex with vertex set $X$.
Since $(X, d) \in {\cal C}_{k, \norm{\cdot}}^m$, there exist $(k-2)$-chains $F = (\vec{f}_1, \ldots, \vec{f}_m)$ such that 
for any $x_1, \ldots, x_k\in X$, $d(x_1, \ldots, x_k) = \norm{\1^T_{x_1,\ldots, x_k}\cdot \delta\cdot F}$ if $x_1, \ldots, x_k$ are distinct and $d(x_1, \ldots, x_k) = 0$ otherwise. 

Therefore, $d$ is non-negative since any norm is non-negative, and $d$ is zero if $x_1, \ldots, x_k$ contains duplicated elements. Hence the properties in \cref{item:strong-non-negative,item:strong-repeats-zero} of a strong pseudo $k$-metric hold.

We next show the symmetry property (\cref{item:weak-permutation-invariance}). Let $\pi : [k] \to [k]$ be a permutation. If $x_1, \ldots, x_k$ are not distinct then $d(x_1, \ldots, x_k) = d(x_{\pi(1)}, \ldots, x_{\pi(k)}) = 0$ by the ``repeated elements property'' (\cref{item:strong-repeats-zero}) of strong pseudo $k$-metrics (which we already showed holds). Otherwise, if $x_1, \ldots, x_k$ are distinct, 
\[
\1_{\pi(x_1,\ldots, x_k)}^T\cdot \delta_{k-2}\cdot\vec{f}_i = \sign(\pi)\cdot \1_{x_1,\ldots, x_k}^T\cdot \delta_{k-2}\cdot\vec{f}_i
\] 
for every $i\in[m]$.  Thus,
\begin{align*}
d(\pi(x_1, \ldots, x_k)) = \norm{
    \1_{x_{\pi(1)},\ldots, x_{\pi(k)})}^T\cdot\delta \cdot F
} = \norm{
    \1_{x_1,\ldots, x_k}^T\cdot\delta \cdot F
} = d(x_1, \ldots, x_k).
\end{align*}

To show strong simplex inequality, fix a $(k-1)$-simplex $t$, and let $\alpha$ be any $(k-1)$-chain such that $\partial_{k-1} \vec{\alpha} = \partial_{k-1} \1_t$, thus $\vec{\alpha}^T\delta_{k-2} = \1_t^T\delta_{k-2}$. We have,
\[
    d(t)
    = \norm{\1_t^T\cdot \delta_{k-2}\cdot F}
    = \norm{\alpha^T\cdot \delta_{k-2}\cdot F}
    \leq |\vec{\alpha}|\cdot (\1_s^T\cdot\delta_{k-2} \cdot F)_{s\in K_{k-1}} = |\vec{\alpha}|\cdot \vec{d},
\]
as desired.
The inequality holds since for any finite set of vectors $(v_1, \ldots, v_k)$, any set of scalars $(\beta_1, \ldots, \beta_k)$, and any norm $\norm{\cdot}$ we have $\norm{\sum_{i=1}^{k}{\beta_i v_i}} \leq \sum_{i=1}^{k}{|\beta_i|\cdot \norm{v_i}}$, by properties \eqref{item:norm-multiplicative} and \eqref{item:norm-triangle} of norms.
\end{proof}

\subsection{Coboundary \texorpdfstring{$k$}{k}-metrics with \texorpdfstring{$p\neq\infty$}{p not infinity}} 
\label{sec:cbd_subset_k_metrics}

The relation between different normed spaces and the expressiveness of one compared to others have been studied extensively. 
See Matousek~\cite[Section 1.5]{matousek13} for a concise summary.  In this section, we generalize a known result of this sort, namely that the space of finite $\ell_{p\neq\infty}$ metric spaces is strictly contained in the space of finite metric spaces: there exist finite metric spaces that are not $\ell_p$ metrics if $p\neq \infty$.  To this end, we define the apex extension of a $k$-metric, which we show is a $(k+1)$-metric.  Intuitively, this $(k+1)$-metric behaves very similarly to the given $k$-metric.  (Yet, this apex extensions are very specific $(k+1)$-metrics, thus we expect the space of $(k+1)$-metrics to be much richer than the space of $k$-metrics.)

\subsubsection{Apex Extension}
\noindent Let $(X, d)$ be a (weak) pseudo $k$-metric, and let $a$ be a new element that we call the \EMPH{apex}, and let $X' = X\cup\{a\}$. We define the function $d'$ on the $(k+1)$-tuples of $X'$, as follows.
\begin{enumerate}[(i)]
    \item For any $x_1,\dots,x_{k+1} \in X'$, if $x_1, \ldots, x_{k+1}$ are not distinct or do not include $a$, it holds that $d'(x_1,\dots,x_{k+1}) = 0$.
    \item Otherwise, if $x_i = a$ for some $i$, $d'(x_1,\dots,x_{k+1}) = d(x_1, \ldots, x_{i-1}, x_{i+1}, \ldots, x_{k+1})$.
\end{enumerate}
We refer to $(X', d')$ as the \EMPH{apex extension} of $(X, d)$ with apex $a$. 
We show that the apex extension gives a $k$-metric.
\begin{lemma}
\label{lem:apex_extension_weak_to_weak}
The apex extension of a (weak) pseudo $k$-metric is a (weak) pseudo $(k+1)$-metric.
\end{lemma}
\begin{proof}
Let $(X, d)$ be a $k$-metric space, and $(X' = X\cup\{a\}, d')$ be its apex extension with apex $a$. 
By the definition of apex extension $d'$ is nonnegative, and $d'(x_1, \ldots, x_{k+1}) = 0$ if $x_1, \ldots, x_{k+1}$ are not distinct. 
Also, by the definition of apex extension, $d'$ is permutation invariant since $d$ is permutation invariant.  Hence, $(X', d')$ has the first three properties of weak pseudo $(k+1)$-metrics.

To show simplex inequality for $(X', d')$, let $x_1,\dots,x_k, x_{k+1}, y\in X'$.
If $x_1, \ldots, x_{k+1}$ are not distinct or they if they do not contain $a$ then $d'(x_1,\ldots,x_{k+1}) = 0$ and the simplex inequality trivially holds. 
Further, if $y\in \{x_1, \ldots, x_{k+1}\}$ again the simplex inequality trivially holds. 
So, we assume that $x_1, \ldots, x_{k+1}$ are distinct and that $x_{k+1} = a$, and $y\neq a$. So,
\begin{align*}
d'(x_1,\ldots,x_{k+1}) &= d'(x_1,\ldots,x_{k},a) \\
&= d(x_1,\ldots,x_{k}) \\
&\leq \sum_{i=0}^{k}{d(x_1,\ldots,x_i,y,x_{i+1},\ldots,x_{k})} \\
&= \sum_{i=0}^{k}{d'(x_1,\ldots,x_i,y,x_{i+1},\ldots,x_{k},a)} \\
&= \sum_{i=0}^{k}{d'(x_1,\ldots,x_i,y,x_{i+1},\ldots,x_{k},a)} + d'(x_1,\ldots,x_{k},y)\\
&= \sum_{i=0}^{k+1}{d'(x_1,\ldots,x_i,y,x_{i+1},\ldots,x_{k},x_{k+1})}.
\end{align*}
The second last equality holds as $d'(x_1,\ldots,x_{k},y) = 0$ because $y\neq a$.
\end{proof}

\begin{figure}[t]
    \centering
    \includegraphics[height=1.2in]{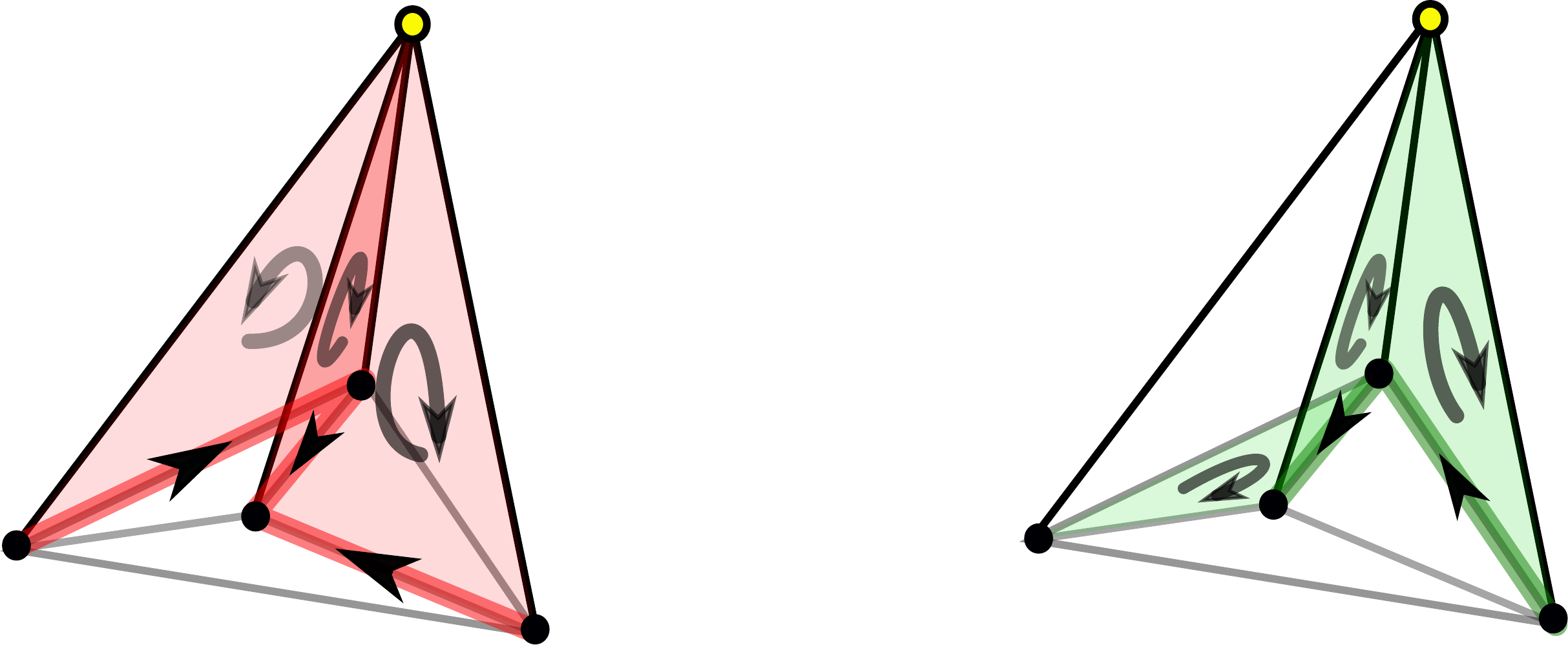}
    \caption{
        A complete graph of black vertices and its apex extension, with the yellow apex.
        Left: a $1$-chain in the graph induced by black vertices is illustrated using thick red edges, and its lift in the $2$-simplicial complex including the apex is illustrated using shaded triangles, Right: a $2$-chain in the $2$-simplicial complex with the apex is illustrated using shaded triangles and its projection in the induced graph of black vertices is illustrated using thick green edges. 
    }
    \label{fig:apex_ext_2}
\end{figure}

Next, we attempt to show that the apex extension of a strong pseudo $k$-metric is a strong pseudo $(k+1)$-metric.
To that end, let $K$ and $K'$ be complete complexes with vertices $X$ and $X'$, respectively.  For any $h\in [|X|-1]$, let $P_h:C_h[K']\rightarrow C_{h-1}[K]$ be the map defined as follows on the basis of $C_h[K']$:
\begin{enumerate}
\item [(1)] for any distinct $x_1, \ldots, x_{h+1}\in X'\backslash\{a\}$, $P_h\cdot \1_{x_1, \ldots, x_{h+1}} = 0$, and 
\item [(2)] for any distinct $x_1, \ldots, x_{h}\in X'\backslash\{a\}$, $P_h\cdot \1_{x_1, \ldots, x_h,a} = \1_{x_1, \ldots, x_h}$.
\end{enumerate}
Note this is a one-to-one map between the $h$-chain space of $K'$ that is spanned by $h$-simplices that include $a$ and the $(h-1)$-chain space of $K$.  We define $L_{h-1} = P_h^T$.  Equivalently, for any distinct $x_1, \ldots, x_{h}\in X$, $L_{h-1}\cdot \1_{x_1\ldots, x_h} = \1_{x_1\ldots x_h a}$. See \cref{fig:apex_ext_2} for examples of the application of the projection and lift operators.
We refer to $P_h$ and $L_h$ as $h$-projection and $h$-lift operators, respectively.
A crucial property of $P_h$ and $L_h$ is stated in the following lemma.
\begin{lemma}
\label{lem:Ph-commutes}
Let $K$ be a complete complex on vertex set $X$, and $K'$ the complete complex on vertex set $X\cup\{a\}$.
For any $h\in [|X|-1]$, let $P_h$ and $L_h$ as defined above.  We have,
\begin{equation}
\label{eqn:P_h-commutes}
\partial_{h-1}[K]\cdot P_h  = P_{h-1} \cdot \partial_{h}[K'],
\end{equation}
and, equivalently,
\begin{equation}
\label{eqn:L_h-commutes}
L_{h-1}\cdot \delta_{h-2}[K]   =  \delta_{h-1}[K']\cdot L_{h-2}.
\end{equation}
\end{lemma}
\begin{proof}
It suffices to show $\partial_{h-1}[K]\cdot P_h \cdot \1_{t'}  = P_{h-1} \cdot \partial_{h}[K']\cdot \1_{t'}$ for every $h$-simplex $(x_1, \ldots, x_{h+1})$ in $K'$.  We consider the following two cases.

If $x_1, \ldots, x_{h+1}$ does not include $a$, then 
$
P_h \cdot \1_{x_1, \ldots, x_{h+1}} = 0,
$
by the definition of $P_h$, thus the left-hand side of \cref{eqn:P_h-commutes} is zero.
Further, none of the $(h-1)$-simplices in $\partial_{h}[K']\cdot \1_{x_1, \ldots, x_{h+1}}$ include $a$, so 
$
P_{h-1}\cdot \partial_{h}[K'] \cdot\1_{x_1, \ldots, x_{h+1}} = 0
$, thus the right-hand side of \cref{eqn:P_h-commutes} is also zero, as desired.

Otherwise, we assume without loss of generality that $x_{k+1} = a$. First, we simplify left-hand side of \cref{eqn:P_h-commutes}
\[
\partial_{h-1}[K]\cdot P_h \cdot\1_{x_1, \ldots, x_{h}, a} = \partial_{h-1}[K]\cdot \1_{x_1, \ldots, x_{h}}
=\sum_{i=1}^{h}{(-1)^{i+1}\1_{x_1, \ldots x_{i-1} x_{i+1} \ldots x_{h}}},
\]
by the definition of $P_h$ and the boundary map.  Then, we simplify right-hand side of \cref{eqn:P_h-commutes}.
\begin{align*}
P_{h-1} \cdot \partial_{h}[K'] \cdot \1_{x_1, \ldots, x_{h}, a}
&= P_{h-1} \cdot \left(
\sum_{i=1}^{h}{(-1)^{i+1}\1_{x_1, \ldots x_{i-1} x_{i+1} \ldots x_{h} a}} + (-1)^{h+2}\1_{x_1, \ldots x_{i-1} x_i x_{i+1} \ldots x_{h}}
\right)  \\ 
&=\sum_{i=1}^{h}{(-1)^{i+1}\1_{x_1, \ldots x_{i-1} x_{i+1} \ldots x_{h}}}.
\end{align*}
So, the left-hand side and right-hand side of \cref{eqn:P_h-commutes} are equal and so the lemma holds.
\end{proof}

\subsubsection{Extending \texorpdfstring{$k$}{k}-metrics to Similar \texorpdfstring{$(k+1)$}{(k+1)}-metrics}

We now show that the apex extension of a coboundary metric is a coboundary metric with respect to the same norm. 
This result together with a later result of this paper will imply that the apex extensions of a strong pseudo $k$-metric is a strong pseudo $(k+1)$-metric.
To that end, we need the following auxiliary lemma.
\begin{lemma}
\label{lem:zero_tree_assumption}
Let $K$ be a simplicial complex, and let $k\in\Z^+$. Also, let $\Lambda$ be a subset of $(k-1)$-simplices of $K$ such that no nonzero linear combination of them is a $(k-1)$-cycle.
Then, for any coboundary $k$-chain $\vec{h}$, there exists a $(k-1)$-chain $\vec{f}$ in $K$, such that
\begin{enumerate}
    \item [(1)] $\vec{h} = \delta_{k-1} \cdot \vec{f}$, and
    \item [(2)] $\vec{f}$ is zero on all simplices in $\Lambda$.
\end{enumerate}
\end{lemma}
\begin{proof}
Since $\vec{h}$ is a coboundary $k$-chain, there exists a $(k-1)$-chains $\vec{f}'$, such that $\vec{h} = \delta_{k-1} \cdot \vec{f}'$. Let $\vec{f}''$  be the \emph{coboundary} $(k-1)$-chain that agrees with $\vec{f}'$ on the values assigned to the simplices of $\Lambda$.  To see that such an $\vec{f}''$ exists, let $K_{\Lambda}$ be the simplicial complex whose set of $(k-1)$-simplices is $\Lambda$, and whose set of $(k-2)$-simplices agrees with $K$. Since $\Lambda$ has no cycle, $\vec{f}''$ is a coboundary chain in $K_{\Lambda}$, thus, there is $(k-2)$-chain $\vec{\beta}$ in $L_{\Lambda}$ such that $\delta_{k-2}[K_\Lambda]\cdot \vec{\beta}$ is equal to $\vec{f}''$ restricted to $K_\Lambda$.
But,  $\vec{f}'' = \delta_{k-2}[K] \cdot \vec{\beta}$ assigns the same set of values to the simplices in $\Lambda$, so $\1_{\lambda}^T\cdot \vec{f}' = \1_{\lambda}^T\cdot \vec{f}''$, for any $\lambda \in \Lambda$.

Now let $\vec{f} = \vec{f}' - \vec{f}''$.  Since $\vec{f}'$ and $\vec{f}''$ agree on the values assigned to simplices of $\Lambda$, $\vec{f}$ assigns zero to these simplices. Further, since $\vec{f}''$ is a coboundary $(k-1)$-chain, $\delta_{k-1}\cdot \vec{f}'' = 0$, therefore, $\delta_{k-1} \cdot \vec{f} = \delta_{k-1}\cdot \vec{f}' = \vec{h}$.
\end{proof}

Now, we are ready to show that the apex extension of a $k$-metric is a couboundary metric if any only if the $k$-metric itself is a coboundary metric.
\begin{lemma}
\label{lem:apex_extension_coboundary_to_coboundary}
The apex extension of a $k$-metric is in ${\cal C}_{k+1,\norm{\cdot}}^m$ if and only if the $k$-metric is in ${\cal C}_{k, \norm{\cdot}}^m$.
\end{lemma}

\begin{proof}
Let $(X, d)$ be a $k$-metric, and let $(X'=\{X\cup\{a\}\}, d')$ be its apex extension.  Also, let $K$ and $K'$ be the complete $(k-1)$- and $k$-simplicial complexes with vertex sets $X$ and $X'$, respectively.

First, we show $(X, d)\in {\cal C}_{k, \norm{\cdot}}^m \Rightarrow (X', d')\in {\cal C}_{k+1, \norm{\cdot}}^m$. 
Since $(X, d)\in {\cal C}_{k, \norm{\cdot}}^m$, there are $(k-2)$-chains $F = (\vec{f}_1, \ldots, \vec{f}_m)$ in $K$ such that for any $t\in K_{k-1}$, $d(t) = \norm{\1_t^T\cdot \delta_{k-2}[K]\cdot F}$.  Let $F' = L_{k-2} F$, where $L_{k-2}$ is the lift operator defined above.
We show, for any $t'\in K'_k$, $d'(t') = \norm{\1_{t'}^T\cdot\delta_{k-1}[K']\cdot F'}$, hence $(X', d')\in {\cal C}_{k+1, \norm{\cdot}}^m$. 
If $t'$ does not include $a$, then $F' = L_{k-2} F = 0$ by the definition of $L_{k-2}$.  Therefore,
$
\norm{\1_{t'}^T\cdot\delta_{k-1}[K']\cdot F'} = 0 = d'(t')
$.
Otherwise, assume without loss of generality that $t' = (x_1, \ldots, x_k, a)\in K'_k$, and let $t = (x_1, \ldots, x_k)\in K'_{k-1}$, so $d'(t') = d(t)$.  Also, we know $d(t) = \norm{\1_t^T \cdot\delta_{k-2}[K]\cdot F}$.  Finally,
\[
\1_{t'}^T\cdot \delta_{k-1}[K']\cdot F' = \1_{t'}^T\cdot \delta_{k-1}[K']\cdot L_{k-2}\cdot F = \1_{t'}^T \cdot L_{k-1} \cdot \delta_{k-2}[K] \cdot F = \1_{t}^T \cdot \delta_{k-2}[K] \cdot F,
\]
where the second equality holds by \cref{lem:Ph-commutes}, and $\1_{t'}^T \cdot L_{k-1} = \1_t$ as $L_{k-1}$ maps $\1_t$ to $\1_{t'}$ by its definition. Therefore, $\norm{\1_{t'}\cdot \delta_{k-1}[K']\cdot F'} = \norm{\1_t^T\cdot \delta_{k-2}[K] \cdot F} = d(t) = d'(t')$, as needed.\\

Second, we show $(X', d')\in {\cal C}_{k+1, \norm{\cdot}}^m \Rightarrow (X, d)\in {\cal C}_{k, \norm{\cdot}}^m$.
Since $(X', d')\in{\cal C}_{k+1, \norm{\cdot}}^m$, there are $(k-1)$-chains $F'=(\vec{f}'_1, \ldots, \vec{f}'_m)$ in $K'$ such that for any $t'\in K'_k$, $d'(t') = \norm{\1_{t'}^T\cdot \delta_{k-1}[K']\cdot F'}$. 
Let $\Lambda = \{\lambda_1, \ldots, \lambda_q\}$ be the set of all $k$-simplices in $K'$ that contain $a$, and consider any nonzero linear combination $\beta = \sum_{i=1}^{q}{c_i{\1_{\lambda_i}}}$.  Since the linear combination is nonzero there exists $1\leq j\leq q$ such that $c_j \neq 0$. Let $\lambda_j = (y_1, \ldots, y_k, a)$, and note that $\1_{y_1, \ldots, y_k}^T\cdot \partial_k[K']\cdot \beta$ is either $c_j$ or $-c_j$, as coefficients for the simplex $(y_1, \ldots, y_k)$ in $\beta$ can only be generated from applying the boundary operator to $\lambda_j$ and not any other $\lambda_i$.  Therefore, $\partial_k[K']\cdot \beta \neq 0$, $\beta$ is not a cycle.  Thus, by \cref{lem:zero_tree_assumption}, we can assume that $F'$ is zero on all simplices in $\Lambda$. Let $F'' = (\vec{f}''_1, \ldots, \vec{f}''_m)$ be the restriction of $F'$ to the $(k-1)$-simplices in $K$, that is for each $i\in[m]$, and each $(k-1)$-simplex $t=(x_1, \ldots, x_k)\in K_{k-1}$, $\1_t^T\cdot \vec{f}''_i = \1_t^T\cdot \vec{f}'_i$. So,
\[
\norm{\1_{x_1, \ldots, x_k}^T\cdot F''} =
\norm{\1_{x_1, \ldots, x_k}^T\cdot F'}  = \norm{\1_{x_1,\ldots, x_{k} a}^T\cdot \delta_{k-1}[K']\cdot F'} = 
d'(x_1,\ldots, x_{k}, a) = d(x_1,\ldots, x_{k}),
\]
where the second equality holds because $F'$ is zero on simplices that contains $a$.
To finish the proof, we show that $\vec{f}''_i$ are coboundary chains in $K$.
To that end, observe for every $x_1, \ldots, x_{k+1}\in X$,
\[
\norm{\1_{x_1, \ldots, x_{k+1}}^T\cdot \delta_{k-1}[K]\cdot F''} =
\norm{\1_{x_1, \ldots, x_{k+1}}^T\cdot \delta_{k-1}[K']\cdot F'}  = 
d'(x_1,\ldots, x_{k+1}) = 0,
\]
by the definition of $d'$.  So, $\delta_{k-1}[K]\cdot F'' = 0$, in particular $\delta_{k-1}[K]\cdot \vec{f}''_i = 0$ for any $i\in [m]$.  Therefore, $\vec{f}''_i$ is a coboundary chain.
\end{proof}

This lemma is powerful in generalizing non-inclusion results from regular metric spaces.  
Specifically, consider ${\cal C}_{2, p}$ and ${\cal C}_{2, q}$ that are the spaces of $\ell_p$ and $\ell_q$ metrics. Given a metric space that belongs to ${\cal C}_{2, p}$ but not ${\cal C}_{2, q}$, one can build a metric space that belongs to ${\cal C}_{k, p}$ but not ${\cal C}_{k, q}$, for $k > 2$ using induction and the lemma above.  In particular, one can generalize the fact that for $p\neq \infty$, the $\ell_p$-metrics does not include all metrics (see Matousek~\cite[Section 1.5-(iv)]{matousek13}).

\begin{corollary}
\label{cor:non_coboundary_metrics_exist}
For every $p\neq \infty$ and $k\geq 2$, there exists a strong $k$-metric space that is not in ${\cal C}_{k, p}$.
\end{corollary}

\subsection{Fr\'{e}chet Embedding for \texorpdfstring{$k$}{k}-metrics}
\label{sec:frechet_embedding}
Contrary to other $\ell_p$ spaces, it is well-known that $\ell_\infty$ contains all finite metric spaces, a result obtained via Fr\'{e}chet embedding. In this section, we show that this result generalizes to strong $k$-metrics.  More specifically, we show that any strong $n$-point $k$-metric belongs to ${\cal C}_{k, \infty}^m$ where $m = \binom{n}{k}$.

The embedding works as follows.
Let $(X, d)$ be a strong pseudo $k$-metric. For each $k$-tuple $x_1, \ldots, x_k\in X$, we show there exists a $k$-metric in ${\cal C}_{k, \infty}^1$ that preserves the value of $d$ for this specific tuple $x_1, \ldots, x_k$, and does not expand the value of $d$ for any other $k$-tuples. This is the main technical part for obtaining our embedding. Our construction is different from the well-known explicit construction of the Fr\'echet embedding for regular metrics. In particular, we build a linear program and show that its (optimal) solution has the properties that we need.

\begin{lemma}
\label{lem:contracting_3_meteric_embedding}
Let $(X, d)$ be a strong pseudo $k$-metric, and let $x_1,\ldots,x_k\in X$.  There exists a coboundary $k$-metric $(X, d')\in {\cal C}_{k, \infty}^1$ such that
\begin{enumerate}[(1)]
\item \label{item:non-expansion} (Non-expansion.) For any $y_1, \ldots, y_k \in X$, $d'(y_1, \ldots, y_k) \leq d(y_1, \ldots, y_k)$, and
\item \label{item:one-dist-preservation} (One-tuple distance preservation.) $d'(x_1,\ldots,x_k) = d(x_1,\ldots,x_k)$.
\end{enumerate}
\end{lemma}
\begin{proof}
Let $K$ be the complete $(k-1)$-simplicial complex with vertex set $X$.  We show how to compute a $(k-2)$-chain $\vec{f}$ such that $\vec{d}' = |\delta_{k-2}\cdot\vec{f}|$ has the required properties of this lemma.

To satisfy \cref{item:non-expansion} of the lemma, we need 
\[
|\delta_{k-2}\cdot\vec{f}| \leq \vec{d} \Longleftrightarrow 
-\vec{d} \leq \delta_{k-2}\cdot\vec{f} \leq \vec{d}.
\]
Here, we have two parallel hyperplane for each $(k-1)$-simplex, so we get $\binom{n}{k}$ non-parallel constraints.
Note that any two $(k-2)$-chains $\vec{f}$ and $\vec{f'}$ with the same cyclic part realize the same metrics as $\delta_{k-2}\cdot(\vec{f} - \vec{f'}) = 0$.
So, we can further reduce our search space for $\vec{f}$ that give rise to metrics with \cref{item:non-expansion}, by adding the extra constraint that $\vec{f}$ is a cocycle (it is orthogonal to the cycle space, which is equal to the boundary space in a complete complex), that is $\partial_{k-2} \vec{f} = 0$.
These two sets of constraints together specify a polytope $P$ within the space of all $(k-2)$-chains.
Since, $\vec{d}$ is non-negative and the boundary of the zero vector is the zero vector, $P$ contains the zero vector, so it is not empty. 
To obtain an $\vec{f}$ that satisfies \cref{item:one-dist-preservation} as well, we solve the following linear program for $\vec{f}$ in the space of $(k-2)$-chains in $K$.
\begin{align*}
&\max~\left(\1_{t}^T\cdot\delta_{k-2}\cdot \vec{f}\right) \\
& \text{s.t.~}-\vec{d} \leq \delta_{k-2}\cdot \vec{f} \leq \vec{d} \\
&~~~~~~~\partial_{k-2}\cdot \vec{f} = 0,
\end{align*}
where $t = (x_1, \ldots, x_k)$.
Intuitively, we would like to maximize the coboundary value on $t$ while preserving \cref{item:non-expansion} for all simplices.  We show that a certain point in $P$ that maximizes the coboundary value at $t$ satisfies \cref{item:one-dist-preservation}.  This is sufficient to prove the lemma as every point in $P$ satisfies \cref{item:non-expansion}.

Let $\vec{f}$ be an optimal \emph{vertex} solution of our linear program; such a vertex solution exists since any vector in the kernel of $\partial_{k-2}$ is a linear combination of coboundary vectors.
Consider a maximal number of independent conditions in the linear program that are tight at $\vec{f}$, suppose these conditions hold at $(k-2)$-simplices $\{t_1, t_2, \ldots, t_r\}$, for which at least one of the corresponding conditions in the linear program is tight.  Note that each of these simplices contributes two conditions to the linear program, namely for each $i\in[r]$, we have $\1_{t_i}^T \cdot \delta_{k-2}\cdot \vec{f}\leq d(t_i)$ and $-\1_{t_i}^T\cdot \delta_{k-2}\cdot \vec{f}\leq d(t_i)$.  Without loss of generality we can assume that all of our tight conditions are of the first type, if not we change the orientation of the simplices to make it so.

Note $\1_t^T\cdot(\delta_{k-2}\cdot\vec{f}) = (\partial_{k-1}\cdot\1_t)^T\vec{f}$. Therefore, $\partial_{k-1}\cdot\1_t$ is the objective vector of our linear program. Also, since $\partial_{k-2}\partial_{k-1} = \vec{0}$, the vector $\partial_{k-1}\cdot\1_t$ and the vectors $\partial_{k-1}\cdot \1_{t_i}$ for $i\in[r]$ are in the subspace specified by $\partial_{k-2} \vec{f} = \vec{0}$.
Since $\vec{f}$ is optimal, $\partial_{k-1}\cdot\1_t$ is in the cone defined by $\partial_{k-1}\cdot \1_{t_i}$ for $i\in[r]$.  So, $\partial_{k-2} \cdot \1_t$ is a non-negative linear combination of these vectors, that is
\begin{equation}
\label{eqn:partial_expand}
\partial_{k-1}\cdot \1_t
= \sum_{i = 1}^{r}{\beta_i (\partial_{k-1}\cdot \1_{t_i})}
= \partial_{k-1}\cdot\left(\sum_{i = 1}^{r}{\beta_i \1_{t_i}}\right),
\end{equation}
with all non-negative $\beta_i$'s.  Since $d$ is a strong $k$-metrics, we have
\begin{equation}
\label{eqn:dt_leq}
d(t) = \1_t \cdot \vec{d} \leq \sum_{i = 1}^{r} |\beta_i| \cdot d(t_i) = \sum_{i = 1}^{r} \beta_i \cdot d(t_i) \ \text{.}
\end{equation}
On the other hand,
\begin{equation}
\label{eqn:dt_geq}
(\partial_{k-1}\cdot \1_t)^T\vec{f} = \1_t^T\cdot(\delta_{k-2}\cdot \vec{f}) \leq d(t)
\end{equation}
as $\vec{f}$ is a feasible solution of the linear program, and that for all $i\in[r]$,
\begin{equation}
\label{eqn:tight_dts}
(\partial\cdot \1_{t_i})^T\vec{f} = d(t_i)
\end{equation}
because these are tight conditions.  So, 
\begin{align*}
\sum_{i = 1}^{r}{\beta_i d(t_i)} &\geq d(t) &&\text{(see \cref{eqn:dt_leq})} \\
&\geq (\partial_{k-1}\cdot \1_t)^T\vec{f} &&\text{(see \cref{eqn:dt_geq})}\\
&= (\sum_{i = 1}^{r}{\beta_i (\partial_{k-1}\cdot \1_{t_i})})^T\vec{f} &&\text{(see \cref{eqn:partial_expand})}\\
&= \sum_{i = 1}^{r}{\beta_i \left((\partial_{k-1}\cdot \1_{t_i})^T\vec{f}\right)} &&\text{}\\
&= \sum_{i = 1}^{r}{\beta_i d(t_i)}
&&\text{(see \cref{eqn:tight_dts})}
\end{align*}
But, since the first and last term are the same, the inequalities in the middle have to be equalities. Therefore, $d(t) = (\partial_{k-1}\cdot \1_t)^T\vec{f}$ as desired.
\end{proof}

Now, we are ready to prove the main theorem of this section.  Specifically, for $k=2$, our theorem implies that any $n$-point metric space is isometrically embeddable into $\ell_\infty^m$ with $m = \binom{n}{2}$, which is similar to the Fr\'echet embedding, but embeds into $m = O(n^2)$ dimensions rather than $m = O(n)$. 
\begin{theorem}
\label{thm:frechet_embedding}
For any integer $k \geq 2$, any $n$-point strong pseudo $k$-metric belongs to ${\cal C}_{k, \infty}^{\binom{n}{k}}$. Therefore, ${\cal C}_{k, \infty} = \S_k$.
\end{theorem}
\begin{proof}
Let $(X, d)$ be a strong pseudo $k$-metric, and let $K$ be a complete simplex with vertex set $X$.
By \cref{lem:contracting_3_meteric_embedding}, for any $(k-1)$-simplex $\lambda$ in $K$, there exists a coboundary $k$-metric $(X, d_\lambda)\in{\cal C}_{k, \infty}^1$ such that (1) $d_\lambda(t)\leq d(t)$ for any $(k-1)$-simplex $t$, and (2) $d_\lambda(\lambda) = d(\lambda)$.  Since $(X, d_\lambda)\in{\cal C}_{k, \infty}^1$, there exists $(k-2)$-chain $\vec{f}_\lambda$ such that $\vec{d}_\lambda = \norm{\delta\cdot \vec{f}_\lambda}_\infty$.

Now, let $F = \left(\vec{f}_\lambda\right)_{\lambda\in K_{k-1}}$ be a matrix specified by its $\binom{n}{k}$ columns, and let for any $(k-1)$-simplex $t$, $d'(t) = \norm{\1_t^T\cdot\delta_{k-2}\cdot F}_\infty$.
We show that $d = d'$. For any $(k-1)$-simplex $t$, we have
\[
\1_t^T\cdot \delta_{k-1}\cdot F = \left(
    \1^T_t\cdot \delta_{k-1}\cdot\vec{f}_\lambda
\right)_{\lambda\in K_{k-1}},
\]
therefore,
\[
d'(t) = \norm{
    \1_t^T\cdot \delta_{k-1}\cdot F}_\infty = \max_{\lambda\in K_{k-1}}{\left(
    \1_t^T\cdot \delta_{k-1}\cdot\vec{f}_\lambda
\right)} = \max_{\lambda\in K_{k-2}}{\left(
    d_\lambda(t)
\right)}.
\]
But, we know that $d_\lambda(t) \leq d(t)$ for all $\lambda\in K_{k-1}$, and that $d_t(t) = d(t)$.  Thus, 
$
\max_{\lambda\in K_{k-2}}{\left(
    d_\lambda(t)
\right)} = d(t).
$
Hence, $d' = d$, as desired.
\end{proof}

We now get two corollaries of \cref{thm:frechet_embedding} (and in particular its implication that ${\cal C}_{k,\infty} = {\cal S}_k$).
First, by combining \cref{thm:frechet_embedding} with the guarantee about apex extensions of coboundary $k$-metrics in \cref{lem:apex_extension_coboundary_to_coboundary} we get the following.
\begin{corollary}
\label{cor:apex_extension_strong_iff_strong}
For every integer $k \geq 2$, the apex extension of a $k$-metric is in ${\cal S}_{k+1}$ if and only if the $k$-metric is in ${\cal S}_{k}$.    
\end{corollary}

Second, by \cref{lem:weak-strong-3-metric} there exists a weak $3$-metric that is not a strong $3$-metric.
Combining this with \cref{lem:apex_extension_weak_to_weak,cor:apex_extension_strong_iff_strong} and using an induction argument we get the following.

\begin{corollary} \label{cor:weak-k-not-strong-k}
For every integer $k \geq 3$, there exists a (weak) pseudo $k$-metric space that is not a strong pseudo $k$-metric space, i.e., $\S_{k} \subsetneq \W_k$.
\end{corollary}

We remark that \cref{cor:weak-k-not-strong-k} shows that the restriction of the Fr\'{e}chet embedding result in~\cref{thm:frechet_embedding} to strong $k$-metrics (and not all weak $k$-metrics) is necessary and not just an artifact of our proof techniques. We again emphasize this as motivation for the definition of strong $k$-metrics.

\subsection{Relating Coboundary \texorpdfstring{$k$}{k}-metrics using Norm Embeddings}
\label{sec:norm-embeddings-to-coboundary}

\color{black}
We now give a technique for generalizing linear norm embeddings to embeddings between coboundary $k$-metric spaces.
Somewhat more precisely, the following proposition says that if there exists a (family of) linear embeddings from $(\R^{m}, \norm{\cdot}_p)$ to $(\R^{m'}, \norm{\cdot}_q)$ that are norm-preserving to within a $1 \pm \eps$ factor on a fixed vector with good probability, then there exists an embedding of $\C_{k, p}^m$ into $\C_{k, q}^{m'}$ that preserves $k$-distances to with a factor of $1 \pm \eps$. 
\begin{proposition} \label{prop:norm-embeddings-coboundary}
Let $k, m, m', n \in \Z^+$, $p, q \in [1, \infty]$, and $\eps \geq 0$. Suppose that there exists a distribution $\mathcal{D}$ of linear maps $R \in \R^{m' \times m}$ such that for every $\vec{x} \in \R^m$,
\begin{equation} \label{eq:norm-embedding-prob}
\Pr[
    (1-\eps) \cdot \norm{\vec{x}}_p \leq \norm{R\cdot \vec{x}}_q \leq (1+\eps) \cdot \norm{\vec{x}}_p] 
    \geq 1 - \eta
\end{equation}
for $\eta = \eta(\eps, k, m, m', n) < 1/\binom{n}{k}$.
Then for any $n$-point $C_{k, p}^m$ space $(X, d)$ there exists an $n$-point $\C_{k, q}^{m'}$ space $(X, d')$ such that for any $t = (x_1, \ldots, x_k) \in X^k$,
\begin{equation} \label{eq:coboundary-embedding-distortion}
(1-\eps) \cdot d(t) \leq d'(t) \leq (1+\eps) \cdot d(t) \ \text{.}
\end{equation}
\end{proposition}

\begin{proof}
By definition of $\C_{k,p}^m$, there exists a matrix $F = (\vec{f}_1, \ldots, \vec{f}_m) \in \R^{\binom{n}{k-1} \times m}$ that induces $(X, d)$. I.e., $F$ is a matrix whose rows are indexed by simplices $t = (x_1, \ldots, x_k)$ and whose columns $\vec{f}_i$ are $(k-2)$-chains such that $d(t) = \norm{\vec{y}_t}_p$ for $\vec{y}_t := (\vec{e}_t^T \delta_{k-2} F)^T$.
Sample $R \sim \mathcal{D}$ and let $F' := F R^T \in \R^{\binom{n}{k-1} \times m'}$.
We claim that with positive probability $F'$ induces a $\C_{k,q}^{m'}$ metric $(X, d')$ satisfying the condition in \cref{eq:coboundary-embedding-distortion} for all $t = (x_1, \ldots, x_k) \in X^k$, and note that this claim implies the proposition.

We now prove the claim. 
By definition, $(X, d')$ satisfies $d'(t) = \norm{R \cdot \vec{y}_t}_q = \norm{(\vec{e}_t^T \delta_{k-2} F')^T}_q$. Furthermore, by \cref{eq:norm-embedding-prob} we have that for any fixed $(k-1)$-simplex $t = (x_1, \ldots, x_k) \in X^k$,
\[
\Pr_{R \sim \mathcal{D}}[
    (1-\eps) \cdot \norm{\vec{y}_t}_p 
    \leq \norm{R\cdot \vec{y}_t}_q 
    \leq (1+\eps) \cdot \norm{\vec{y}_t}_p] 
    \geq 1 - \eta \ \text{.}
\]
The claim follows by taking a union bound over all $\binom{n}{k}$ such simplices $t$ consisting of distinct points in $X$ with standard orientation and using the fact that $\eta < 1/\binom{n}{k}$.

\end{proof}

We next give two corollaries of \cref{prop:norm-embeddings-coboundary} corresponding to the Johnson-Lindenstrauss lemma for dimension reduction in $\ell_2$ and norm embeddings from $\ell_2$ to $\ell_p$.
We start by giving the version of the Johnson-Lindenstrauss lemma (dimension reduction using random Gaussian projections) given in Matou\v{s}ek~\cite[Lemma 2.3.1]{matousek13}.
\begin{theorem}[Johnson-Lindenstrauss lemma]
\label{thm:jl-lemma}
Let $m, m' \in \Z^+$, and $\varepsilon\in(0,1)$. 
There exists an (efficiently sampleable) distribution of linear maps $R:\R^m\rightarrow \R^{m'}$ such that for every $\vec{x} \in \R^m$ we have
\[
\Pr[
    (1-\eps) \cdot \norm{\vec{x}}_2 \leq \norm{R\cdot \vec{x}}_2 \leq (1+\eps) \cdot \norm{\vec{x}}_2]\geq 1 - 2^{-c\varepsilon^2 m'} \ \text{,}
\]
where $c>0$ is a constant.
\end{theorem}

From this, we get an analog of the Johnson-Lindenstrauss lemma for coboundary $k$-metric spaces.

\begin{corollary} \label{cor:jl-coboundary-metrics}
For every $n$-point $\C_{k,2}^m$ space $(X, d)$ and $\eps \in (0, 1)$, there exists a $\C_{k,2}^{m'}$ space $(X, d')$ such that
\[
(1-\eps) \cdot d(t) \leq d'(t) \leq (1+\eps) \cdot d(t) \ \text{,}
\]
where $m' = c' k \log n/\eps^2$ for some absolute constant $c' > 0$.
\end{corollary}

\begin{proof}
The corollary follows by combining \cref{prop:norm-embeddings-coboundary,thm:jl-lemma}.
In particular, we note that $\binom{n}{k} \leq n^k$, and so for a sufficiently large constant $c' > 0$, $2^{-c \eps^2 m'} = 2^{-c c' \log(n^k)} < 1/\binom{n}{k}$.
\end{proof}

We next give a result about embedding $\ell_2$ spaces into $\ell_p$ spaces from~\cite{FigielLM77}, building on work of~\cite{Dvoretzky60}. See also~\cite[Theorem 3.2]{conf/stoc/RegevR06}.

\begin{theorem}[{Embedding $\ell_2$ into $\ell_p$;~\cite{FigielLM77}}] \label{thm:l2-into-lp}
Let $p$ be a real number satisfying $1 \leq p < \infty$, let $m \in \Z^+$, and let $\eps > 0$. There exists a distribution $\mathcal{D}$ of (efficiently sampleable) linear maps $R \in \R^{m' \times m}$ such that with probability at least $1 - 2^{-cm}$ it holds that for all $\vec{x} \in \R^m$,
\[
(1 - \eps) \cdot \norm{\vec{x}}_2 \leq \norm{R \cdot \vec{x}}_p \leq (1 + \eps) \cdot \norm{\vec{x}}_2 \ \text{,}
\]
where $c > 0$ is a constant and $m'$ is the smallest integer satisfying
\[
m' \geq
\begin{cases}
\frac{m}{\eps^2} & \text{if $1 \leq p < 2$ ,} \\
\big(\frac{m}{\eps^2 p}\big)^{p/2} & \text{if $2 \leq p < \infty$ .}
\end{cases}
\]
\end{theorem}

We remark that, unlike \cref{thm:jl-lemma}, with high probability the embedding $R$ in \cref{thm:l2-into-lp} approximately preserves norms of \emph{all} $\vec{x} \in \R^m$. 
That is, the result holds even for infinite $X \subseteq \R^m$ (including $X = \R^m$), and in particular $m'$ has no dependence on $n = \card{X}$ in the finite case.
From \cref{thm:l2-into-lp}, we get a corollary about embedding $\C_{k, 2}^m$ spaces into $\C_{k, p}^{m'}$ spaces.

\begin{corollary} \label{cor:Ck2-Ckp}
Let $p$ be a real number satisfying $1 \leq p < \infty$, let $k, m \in \Z^+$, and let $\eps > 0$.
For every $\C_{k,2}^m$ space $(X, d)$ there exists a $\C_{k,p}^{m'}$ space $(X, d')$ such that
\[
(1-\eps) \cdot d(t) \leq d'(t) \leq (1+\eps) \cdot d(t) \ \text{,}
\]
where $c > 0$ is a constant and $m'$ is the smallest integer satisfying
\[
m' \geq
\begin{cases}
\frac{m}{\eps^2} & \text{if $1 \leq p < 2$ ,} \\
\big(\frac{m}{\eps^2 p}\big)^{p/2} & \text{if $2 \leq p < \infty$ .}
\end{cases}
\]
\end{corollary}

\begin{proof}
The corollary follows by combining \cref{prop:norm-embeddings-coboundary,thm:l2-into-lp}. 
\end{proof}

We conclude with a remark about the \emph{(non-)efficiency} of the embedding results we gave between coboundary $k$-metric spaces. First, we note that the known distributions of linear norm embeddings $R$ in \cref{thm:jl-lemma,thm:l2-into-lp} that we used are efficiently sampleable, i.e., it is possible to sample an $R$ meeting the guarantees of the theorems in time $\poly(n, m)$ (for fixed $\eps$ and $p$). However, in order to make the embedding result in \cref{prop:norm-embeddings-coboundary} (and those in \cref{cor:jl-coboundary-metrics,cor:Ck2-Ckp}) efficient, it seems necessary not just to know $(X, d)$, but to know a matrix $F$ that induces $(X, d)$.
We leave the corresponding (inverse) problem---of efficiently computing such an $F$ given a space $(X, d)$ promised to be a $\C_{k, p}^m$ space as input---as an intriguing open question.

\subsection{Hypertree \texorpdfstring{$k$}{k}-metrics}
\label{sec:tree_metrics}
In this section, we study minimum bounding chain $k$-metrics and hypertree $k$-metrics, generalizing shortest path metrics and tree metrics of graphs. Our main result in this section is a generalization of the well-known fact that any tree metric is an $\ell_1$-metric~\cite{IndykMatousek17Handbook}.

\paragraph{Minimum bounding chain $k$-metrics.}
Let $k \geq 2$ be an integer, and let $K$ be a $(k-1)$-simplicial complex with vertex set $X$ that has complete $(k-2)$-skeleton, in which all $(k-2)$-cycles are boundary cycles. Also, let $w$ be a set of non-negative weights on $(k-1)$-simplices of $K$.
Let $(X, d)$ be the $k$-metric defined as follows.  For any $x_1, \ldots, x_k\in X$,
\[
d(x_1, \ldots, x_k) = \min_{\vec{\alpha}}{\left(|\vec{\alpha}|\cdot\vec{w}\right)}
\]
where the minimum ranges over all $(k-1)$-chains $\vec{\alpha}$ such that $\partial_{k-1}\cdot\alpha = \sum_{i=1}^{k}{(-1)^{i+1}\1_{x_1 x_2 \ldots x_{i-1} x_{i+1} \ldots x_k}}$.  In this case, we call $(X, d)$ the minimum bounding chain $k$-metric of $K$. It is straightforward to check that $(X, d)$ is a strong $k$-metric. The shortest paths metric of graphs is the special case of minimum bounding chain $k$-metrics for $k=2$.

We study the minimum bounding chain metric for the case that $K$ has no $(k-1)$-cycle, in addition to the fact that all its $(k-2)$-cycles are boundary cycles.
In this case, we call $K$ a \EMPH{$(k-1)$-hypertree} and $(X, d)$ a \EMPH{hypertree $k$-metric}.  In particular, trees in graph theory are $1$-hypertrees and tree metrics are hypertree $2$-metrics.  We define ${\cal T}_k$ to be the space of all finite hypertree $k$-metrics.
It is well-known that any tree metric isometrically embeds into $\ell_1$~\cite{IndykMatousek17Handbook}.  We next generalize this result to hypertree $k$-metrics.
\begin{theorem}\label{thm:tree_metrics_are_l1}
Any hypertree $k$-metric is a coboundary $k$-metric with respect to the $\ell_1$-norm; i.e.,~${\cal T}_k \subseteq {\cal C}_{k,1}$.
\end{theorem}
\begin{proof}
Let $(X, d)$ be a hypertree $k$-metric: the minimum bounding chain $k$-metric of a $(k-1)$-simplicial complex $T$ with vertex set $X$, complete $(k-2)$-skeleton, and no $(k-1)$-cycle, in which all $(k-2)$-cycles are boundary cycles. Also, let $w$ be non-negative weights on the $(k-1)$-simplices of $T$.
Further, let $K$ be the complete complex with vertex set $X$.

Let $\vec{\alpha}$ be the unique $(k-1)$-chain in $T$ such that $\partial_{k-1}[T]\cdot\vec{\alpha} = \partial_{k-1}[K]\cdot \1_t$, or equivalently, $\vec{\alpha}^T\cdot \delta_{k-2}[T] = \1_t^T\cdot \delta_{k-2}[K]$
\footnote{Note that the two spaces $C_{k-2}[K]$ and $C_{k-2}[T]$ are isomorphic as $T$ and $K$ both have the same vertex set and complete $(k-2)$-skeleton.  We implicitly use this isomorphism when we compare $(k-2)$-chains in $K$ and $T$}.
Such an $\vec{\alpha}$ exists since all $(k-2)$-cycles in $T$ are boundary cycles, and is unique since $T$ has no $(k-1)$-cycle. Therefore, by the definition of the minimum bounding chain $k$-metrics, $d = |\vec{\alpha}|\cdot \vec{w}$.

Let $W = \diag(\vec{w})$, and observe 
$|\vec{\alpha}|\cdot \vec{w} = \norm{\vec{\alpha}^T\cdot W}_1$.  Each column of $W$ is a $(k-1)$-chain, and $T$ does not have a $(k-1)$-cycle. Therefore, each column of $W$ is a coboundary $(k-1)$-chain.  So, we can write $W = \delta_{k-2}\cdot F$ for some matrix $F$ whose columns are $(k-2)$-chains. Overall,
\[
d = \norm{\vec{\alpha}^T\cdot W}_1 = \norm{\vec{\alpha}^T\cdot \delta_{k-2}\cdot F}_1.
\]
So, $(X, d) \in {\cal C}_{k,1}$, as needed.

To find out the dimension, note that the number of columns of $W$ and so $F$ equals the number of $(k-1)$-simplices of $T$, denoted $n_{k-1}$, which is $O(n^{k-1})$.
To see that, note $n_{k-1}$ is equal to the dimension of the $(k-2)$-boundary cycles, which is the same as $(k-2)$-cycles, as in $T$ these two .  This cycle space is the kernel of $\partial_{k-2}$, whose dimensionality is bounded by $O(n^{k-1})$.
\end{proof}

\section{Volume Metrics}
\label{sec:volume_metric}

We study volume metrics as intuitive examples of $k$-metrics.  They are also natural generalization of Euclidean metrics.
Therefore, it makes sense to study them with coboundary metrics under $\ell_2$ norm, the other, probably less immediate, generalization of Euclidean metrics, introduced in this paper.  We show evidence that $\ell_2$ coboundary $k$-metrics are richer than the volume $k$-metrics.  Specifically, we show that volume $k$-metric of points in $\R^m$ belong to coboundary $k$-metrics with the $2$-norm in $O(m^k)$ dimensions (\cref{sec:embedding_volume_to_coboundary_metrics}).  On the other hand, we show a coboundary $n$-point metric in ${\cal C}_{3, 2}^1$ that can not be realized as the volume metric of points in $\Omega(\log n)$ dimensions (\cref{sec:cbd_into_volume_embd}).

\subsection{High-Dimensional Volume}
\label{sec:vol_metric_def}

In this section, we review definitions and useful facts concerning high-dimensional volume.
The $k$-dimensional \EMPH{signed volume} of the $k$-simplex with vertices $x_0, x_1, \ldots x_{k}\in \R^{k}$ is defined as
\begin{equation} \label{eq:signed-vol}
\svol_k(\vec{x}_0, \vec{x}_1, \ldots, \vec{x}_{k}) := \frac{1}{k!} \cdot \det (A) \ \text{,}
\end{equation}
where $A := (\vec{x}_1 - \vec{x}_0, \ldots, \vec{x}_k - \vec{x}_0)$.
In turn, the \EMPH{volume} of this simplex is
\begin{equation} \label{eq:volume}
\vol_k(\vec{x}_0, \vec{x}_1, \ldots, \vec{x}_{k}) := |\svol_k(\vec{x}_0, \vec{x}_1, \ldots, \vec{x}_{k})| \ \text{.}
\end{equation}

More generally, for $x_0, \ldots, x_{k}\in \R^m$ (where $m$ is not necessarily $k + 1$) the volume of the simplex $\conv(x_0, x_1, \ldots, x_k)$ is equal to the square root of the determinant of the Gram matrix of $A := (\vec{x}_1 - \vec{x}_0, \ldots, \vec{x}_k - \vec{x}_0)$:
\begin{equation} \label{eq:gram-vol}
\vol_{k}(\vec{x}_0, \vec{x}_1,\ldots, \vec{x}_k) = \frac{1}{k!}\sqrt{\det(A^T A)} \ \text{.}
\end{equation}
We note that the (signed) volume of the simplex $(x_0, x_1, \ldots, x_k)$ is non-zero if and only if $A$ has full column rank.

One can also compute the volume of a $k$-simplex in $\R^m$ in terms of its projections onto axis-aligned hyperplanes using the \emph{Cauchy-Binet formula}, as we describe below (see, e.g.,~\cite{cauchy-binet} for a proof).
For a matrix $A \in \R^{m \times k}$ and sets $I \subseteq [m]$ and $J \subseteq [k]$, let $A_{I, J} \in \R^{\card{I} \times \card{J}}$ denote the submatrix of $A$ formed by the rows of $A$ indexed by $i \in I$ and columns of $A$ indexed by $j \in J$.

\begin{theorem}[Cauchy-Binet formula] \label{thm:cauchy-binet}
Let $A \in \R^{m \times k}$ for $k, m \in \Z^+$ with $k \leq m$. Then
\[
\det(A^T A) = \sum_{\substack{I \subseteq [m], \\ \card{I} = k}} \det(A^T_{[k], I} \cdot A_{I, [k]}) \ \text{.}
\]
\end{theorem}

Combining \cref{thm:cauchy-binet} with \cref{eq:gram-vol}, we get as a corollary that the squared volume of a $k$-simplex $\sigma \subseteq \R^m$ is equal to the sum of the squared volumes of the $\binom{m}{k}$ projections $\pi_I(\sigma)$ for $I \subseteq [m]$, $\card{I} = k$ of $\sigma$ onto $k$-dimensional axis-aligned hyperplanes.

\begin{corollary} \label{cor:cauchy-binet-volume}
Let $k, m \in \Z^+$ with $m \geq k$, and let $\sigma \subseteq \R^m$ be a $k$-simplex. Then
\begin{equation}
\label{eq:cauchy_bin_volume}
\vol_{k}(\sigma)^2 =
    \sum_{\substack{I \subseteq [m], \\ \card{I} = k}} \vol_k(\pi_I(\sigma))^2
 = \norm{\vol_k(\pi_I(\sigma)_{I \subseteq [m], \card{I} = k})}^2
\ \text{,}
\end{equation}
where $\pi_I$ denotes orthogonal projection onto the axis-aligned hyperplane spanned by the vectors $\1_i$ for $i \in I$.
\end{corollary}

We note in passing that one could also define the ``$(k, p)$-volume'' of a $k$-simplex $\sigma \subseteq \R^m$ as
\[
\nu_{k,p}(\sigma) := \Big(\sum_{\substack{I \subseteq [m], \\ \card{I} = k}} \vol_k(\pi_I(\sigma))^p \Big)^{1/p} \ \text{,}
\]
and that these functions $\nu_{k,p}$ are also strong pseudo $(k+1)$-metrics.
For example, $\nu_{2,1}$ corresponds to adding up the areas of the projections of a triangle embedded in $\R^m$ onto axis-aligned planes.
These functions might also be of interest.

\subsection{Volume \texorpdfstring{$k$}{k}-metrics}
Now, we are ready to define volume metrics and study their relationship with coboundary metrics under the Euclidean norm.

\begin{definition}[Volume $k$-metrics]
\label{def:volume_norms}
Let $X$ be a finite set of points, and let $d$ be a function on $k$-tuples of $X$.
We say that $(X, d)$ is an $m$-dimensional \EMPH{volume metric} if there exists a map $f:X\rightarrow \R^m$, such that for any $k$ points $x_1, \ldots, x_k \in X$, we have 
\[
d(x_1, \ldots, x_k) = \vol_{k-1}(f(x_1), \ldots, f(x_k)).
\]
\end{definition}
\noindent We define the family of all $m$-dimensional volume metrics by ${\cal V}_{k}^m$. Further, we define ${\cal V}_k = \bigcup_{k\in\Z^+}{{\cal V}_k^m}$.

\label{sec:embedding_volume_to_coboundary_metrics}
We can show that if $(X,d)$ is an $m$-dimensional volume $k$-metric then it belongs to ${\cal C}_{k, 2}^{m'}$ where $m' = \binom{m}{k}$.  So, the family of coboundary metrics is equivalent or richer. The following lemma is the main technical step for proving this result.

\begin{lemma} 
\label{lem:volume_to_1d_cbd}
For any integer $k \geq 2$, the $(k-1)$-dimensional volume $k$-metric is a $1$-dimensional coboundary $k$-metric with respect to the $\ell_2$ norm, i.e.,~${\cal V}_k^{k-1}\subseteq {\cal C}_{k, 2}^1$. 
\end{lemma}

\begin{proof}
We show any $(X, d) \in {\cal V}_k^{k-1}$ belongs to ${\cal C}_{k, 2}^{1}$. 
By the definition of volume $k$-metrics, there exists a function $g:X\rightarrow \R^{k-1}$, such that for any $x_1, \ldots, x_k$, $d(x_1, \ldots, x_k) = \vol_{k-1}(g(x_1), \ldots, g(x_k)) = \vol_{k-1}(\vec{y}_1, \ldots, \vec{y}_k)$, where $\vec{y}_i := g(x_i)$ for $i\in [k]$.

Let $K$ be a complete $(k-1)$-dimensional complex with vertex set $X$. For any $(k-2)$-simplex $s = (x_1, \ldots, x_{k-1})$,
\[
f(s) = f((x_1, \ldots, x_{k-1})) := \text{svol}_{k-1}(\vec{0}, \vec{y}_1, \ldots, \vec{y}_k) \ \text{.}
\]

Note, this is a well-defined chain, as swapping two columns changes the sign of the signed volume and the sign of the corresponding oriented simplex.
We show that $|\1_t^T\cdot \delta_{k-2}\cdot \vec{f}| = d(t)$ for every $(k-1)$-simplex $t = (x_1, \ldots, x_k)$ in $K$.
\begin{align*}
\1_t^T\cdot\delta_{k-2} \vec{f} = (\partial_{k-1}\cdot\1_t)^T\cdot \vec{f}
&= (\partial_{k-1}\cdot\1_{x_1,\ldots, x_{k}})^T\cdot \vec{f} \\
&= \left(\sum_{i=1}^{{k}}{(-1)^{i+1}\1^T_{x_1 \ldots x_{i-1} x_{i+1} \ldots x_{k}}}\right) \cdot \vec{f} \\
&= \sum_{i=1}^{{k}}{(-1)^{i+1}f((x_1, \ldots, x_{i-1}, x_{i+1}, \ldots, x_{k}))} \\
&= \sum_{i=1}^{{k}}{
    (-1)^{j+1}\cdot\text{svol}_{k-1}\left(\vec{0}, \vec{y}_1, \ldots, \vec{y}_{j-1}, \vec{y}_{j+1}, \ldots, \vec{y}_{k}\right)
} \\
&= \frac{1}{(k-1)!}\cdot\sum_{i=1}^{{k}}{
    (-1)^{j+1}\cdot\det\left(\vec{y}_1, \ldots, \vec{y}_{j-1}, \vec{y}_{j+1}, \ldots, \vec{y}_{k})\right)
} \\
&= \frac{1}{(k-1)!}\cdot \det \begin{pmatrix}
        1 & 1 &\ldots  & 1 \\
        \vec{y}_1 & \vec{y}_2 & \ldots & \vec{y}_k
    \end{pmatrix}
\end{align*}
The last equality follows from expanding the determinant using the first (i.e., all-$1$s) row of the matrix on the right-hand side.
Furthermore, by subtracting the first column from all other columns, we have
\begin{align*}
\det \begin{pmatrix}
    1 & \ldots & 1  & 1 \\
    \vec{y}_1 &  \ldots & \vec{y}_{k-1} & \vec{y}_k
\end{pmatrix} & = \det \begin{pmatrix}
    1 & 0 & \ldots & 0 \\
    \vec{y}_1 & \vec{y}_2-\vec{y}_1 & \ldots & \vec{y}_k - \vec{y}_1
\end{pmatrix} \\ &= \det\left(
    \vec{y}_2-\vec{y}_1, \ldots, \vec{y}_k-\vec{y}_1
\right) \\ &= (k-1)!\cdot\text{svol}_{k-1}(\vec{y}_1, \ldots, \vec{y}_k).
\end{align*}
So, we conclude $|\1_t^T\cdot\delta \vec{f}| = \text{vol}_{k-1}(\vec{y}_1,\ldots, \vec{y}_k) = d(t)$, as desired.
\end{proof}

We now show how to use \cref{lem:volume_to_1d_cbd}, which shows $\V_{k}^{k-1}\subseteq\C_{k, 2}^1$, to obtain a general embedding from $\V_k^m$ to ${\cal C}_{k, 2}^{m'}$ with $m' = \binom{m}{k-1}$.

\begin{theorem}
\label{thm:volume_to_cbd}
Any $m$-dimensional volume $k$-metric is a $\binom{m}{k-1}$-dimensional coboundary $k$-metric with respect to $\ell_2$-norm, i.e.~${\cal V}_k^m \subseteq {\cal C}_{k, 2}^{m'}$ with $m' = \binom{m}{k-1}$.  In particular, ${\cal V}_k \subseteq {\cal C}_{k, 2}$ and therefore also $\V_k \subseteq \S_k$.
\end{theorem}

\begin{proof}
We show any $(X, d) \in {\cal V}_k^{m}$ belongs to ${\cal C}_{k, 2}^{m'}$. 
By the definition of volume $k$-metrics, there exists a function $g:X\rightarrow \R^{m}$, such that for any $x_1, \ldots, x_k$, $d(x_1, \ldots, x_k) = \vol_{k-1}(g(x_1), \ldots, g(x_k))$. 

Let $\Pi = \{\pi_1, \ldots, \pi_\ell\}$ be the set of $t = \binom{m}{k-1}$ orthogonal projections onto all $(k-1)$-dimensional axis-aligned hyperplanes in $\R^m$.
By \cref{lem:volume_to_1d_cbd}, for any $\pi_i\in \Pi$, there exists a $(k-1)$-chain $f_i$ such that for any $(k-1)$-simplex $t = (x_1, \ldots, x_k)$, $\norm{\1_t^T\cdot \delta \cdot \vec{f}_i} = \text{vol}_k(\pi_i(g(t)))$, which is the $k$th volume of the projection of $g(t) = (g(x_1), \ldots, g(x_k))$ under $\pi_i$.

Now, let $F = (\vec{f}_1, \ldots, \vec{f}_{\ell})$, and observe
\[
\1_t^T\cdot \delta \cdot F
= \1_t^T\cdot \delta \cdot (\vec{f}_1, \ldots, \vec{f}_\ell)
= \left(\1_t^T\cdot \delta \cdot \vec{f}_1, \ldots, \1_t^T\cdot \delta \cdot \vec{f}_\ell\right)
= \left(\vol_{k-1}(\pi_1(g(t)), \ldots, \vol_{k-1}(\pi_\ell(g(t)))\right),
\]
where the last equality follows by \cref{lem:volume_to_1d_cbd}.  Therefore, by \cref{cor:cauchy-binet-volume}, 
\[
\norm{\1_t^T\cdot \delta \cdot F}_2
= \vol_{k-1}((g(t)),
\]
as promised.
\end{proof}

\subsection{Relating Coboundary and Volume \texorpdfstring{$k$}{k}-metrics}
\label{sec:cbd_into_volume_embd}

In this section, we show three results either showing the impossibility of realizing a coboundary metric as a volume metric (i.e., the impossibility of isometrically embedding a coboundary metric space into a volume metric space), or giving a lower bound on the dimension and distortion necessary to embed a coboundary metric space into a volume metric space.
Combined with \cref{thm:volume_to_cbd}, this shows that coboundary metric spaces are \emph{strictly} more expressive than volume metric spaces. These results are formally incomparable.
The two results in \cref{sec:impossibility-coboundary-volume} show \emph{impossibility} results for isometrically embedding $k$-\emph{pseudo}metric spaces, whereas the result in \cref{sec:lower-bound-coboundary-volume} shows a dimension \emph{lower bound} for embedding \emph{non-pseudo} $k$-metric spaces with constant distortion.
(We refer the reader back to \cref{sec:summary-volume} for additional discussion.)

In what follows, we use $\conv(\vec{x}, \vec{y}, \vec{z})$ to refer to the convex hull of three points $\vec{x}, \vec{y}, \vec{z} \in \R^m$, which is simply the triangle that they span. We start with the following elementary claim relating the minimum and maximum side lengths of a triangle with the triangle's area.

\subsubsection{Impossibility Results for Embedding Coboundary \texorpdfstring{$k$}{k}-pseudometrics}
\label{sec:impossibility-coboundary-volume}

We start by showing an impossibility result for realizing a $4$-point $2$-dimensional coboundary $3$-pseudometric space as a volume $3$-pseudometric space.

\begin{theorem} \label{thm:4-point-2d-coboundary}
There exists a $4$-point $\C_{3, 2}^2$ pseudometric space that is not a $\V_3$ pseudometric space, i.e.~$\C_{3,2}\not\subseteq \V_3$.
\end{theorem}

\begin{proof}
Let $X = \set{x_1, x_2, x_3, x_4}$ and define $d$ by $d(x_1, x_2, x_3) = 0$ and $d(x_1, x_2, x_4) = d(x_1, x_3, x_4) = d(x_2, x_3, x_4) = 1$.%
\footnote{It is straightforward to check that $(X, d)$ is a $3$-pseudometric space, but it is not a $3$-metric space because $d(x_1, x_2, x_3) = 0$.}

First, we show that $(X, d)$ is a $\C_{3, 2}^2$ space.  Consider the $3$-simplex spanned by $x_1, x_2, x_3, x_4$ in standard orientation. Label each of the edges $(x_1, x_2)$, $(x_1, x_3)$, and $(x_2, x_3)$ with $(0, 0)^T$ and the edges $(x_1, x_4)$, $(x_2, x_4)$, and $(x_3, x_4)$ with $(0, 0)^T$, $(1, 0)^T$, and $(1/2, \sqrt{3}/2)^T$, respectively. One can check that the coboundary metric induced by the standard orientation of the simplex and these edge labels is equal to $d$.%
\footnote{The labels of the edges $(x_1, x_4)$, $(x_2, x_4)$, and $(x_3, x_4)$ correspod to the three vertices of an equilateral triangle in $\R^2$. One can check that this induces the correct coboundary $3$-pseudometric $d$. For example, $d(x_2, x_3, x_4) = \norm{(0, 0)^T + (1/2, \sqrt{3}/2)^T - (1, 0)^T}_2 = 1$.}

Second, we show that $(X, d)$ is not a $\V_3$ space, i.e., that $(X, d)$ is not a $\V_3^m$ space for any $m$.
To that end, suppose for contradiction that there exist $\vec{x}_1, \vec{x}_2, \vec{x}_3, \vec{x}_4 \in \R^m$ for some $m$ (corresponding to the embedding $x_i \mapsto \vec{x}_i$ for $i = 1, 2, 3, 4$) such that $\vol_2(\conv(\vec{x}_i, \vec{x}_j, \vec{x}_k)) = d(x_i, x_j, x_k)$ for $i, j, k \in \set{1, 2, 3, 4}$.
Because $\vol_2(\conv(\vec{x}_1, \vec{x}_2, \vec{x}_3)) = 0$, we must have that $\vec{x}_1, \vec{x}_2, \vec{x}_3$ are collinear. So, assume without loss of generality that $\vec{x}_1 = (a_1, 0), \vec{x}_2 = (a_2, 0), \vec{x}_3 = (a_3, 0), \vec{x}_4 = (a_4, b) \in \R^2$ for $a_1, a_2, a_3, a_4, b \in \R$. Then using the formula for triangle area, we must have $1/2 \cdot \abs{a_1 - a_2} \cdot b = 1/2 \cdot \abs{a_1 - a_3} \cdot b = 1/2 \cdot \abs{a_2 - a_3} \cdot b = 1$, but this is impossible because no three values $a_1, a_2, a_3 \in \R$ satisfy $\abs{a_1 - a_2} = \abs{a_1 - a_3} = \abs{a_2 - a_3} > 0$.
\end{proof}

We next give an elementary proof that no five points in the plane are such that every three of them spans a unit area triangle (but that there are four points in the plane with this property). This implies that the $4$-point ``discrete $3$-metric space'' from \cref{lem:all1s-coboundary} is realizable as a $\V_3^2$ metric but that the corresponding $5$-point metric space is not.
\begin{lemma} \label{lem:unit-triangles-R2}
Let $X \subsetneq \R^2$ be such that for all distinct triples $\vec{x}, \vec{y}, \vec{z} \in X$, $\vol_2(\conv(\vec{x}, \vec{y}, \vec{z})) = 1$. Then $\card{X} \leq 4$. 
Furthermore, there is such a set $X$ with $\card{X} = 4$ satisfying this property.
\end{lemma}

\begin{proof}
Assume without loss of generality that $X$ is contained in the upper half-plane of $\R^2$, and that $(0, 0)^T, \ell \cdot \vec{e}_1 \in X$ for some $\ell \neq 0$. (One can enforce this assumption by noting that (i) distances are invariant under translation and rotation, and (ii) there exists a pair $\vec{x}, \vec{y} \in X$ such that all points of $X$ reside in one of the closed half planes of the $(\vec{x}, \vec{y})$ line.)
Then because all triples of points in $X$ span triangles with area $1$, all points in $X \setminus \set{(0, 0)^T, \ell \cdot \vec{e}_1}$ must lie on the horizontal line $L = \set{(x, 2/\ell)^T : x \in \R}$.
Moreover, any three collinear points in $\R^2$ span a triangle with area $0$, and so $\card{L \cap X} \leq 2$. It follows that $\card{X} \leq 4$.

On the other hand, let $X$ be the point set of four points that are the vertices of a square with side length $\sqrt{2}$.  It is easy to check that the area of any triangle whose vertices are triples in $X$ is $(\sqrt{2} \cdot \sqrt{2})/2 = 1$.
\end{proof}

We now show that ``the discrete $3$-metric'' $d$ with $d(x_i, x_j, x_k) = 1$ for distinct $x_i, x_j, x_k$ and $d(x_i, x_j, x_k) = 0$ otherwise is a coboundary metric.

\begin{lemma} \label{lem:all1s-coboundary}
For every $n \geq 3$, there is an $n$-point $3$-metric space $(X, d)$ in $\C_{3, |\cdot|}^1$ such that $d(x, y, z) = 1$ for any distinct $x, y, z \in X$.
\end{lemma}

\begin{proof}
Let $X = \set{x_1, \ldots, x_n}$, and let $K$ be the complete $2$-simplicial complex with vertex set $X$.
Let $\alpha$ be the $1$-chain such that $\vec{\alpha}[x_i, x_j] = 1$ for any $1 \leq i < j \leq n$ with induced coboundary metric $d(x_i, x_j, x_k) = (\delta_1 \cdot \vec{\alpha})[x_i, x_j, x_k]$ for $x_i, x_j, x_k\in X$. 
By~\cref{lem:k_coboundary_metrics}, $d$ is a $3$-metric and hence is invariant under permutation of its arguments.
Therefore, it suffices to show that $d(x_i, x_j, x_k) = 1$ for $i < j < k$.
By the definition of the coboundary operator, for such $i, j, k$ we have
\[
d(x_i, x_j, x_k) = \abs{\vec{\alpha}[x_i, x_j] - \vec{\alpha}[x_i, x_k] + \vec{\alpha}[x_j, x_k]} = \abs{1 - 1 + 1} = 1 \ \text{,}
\]
as needed.
\end{proof}

We next show how to combine \cref{lem:unit-triangles-R2,lem:all1s-coboundary} and apex extension to get an impossibility result for isometrically embedding coboundary metric spaces into volume metric spaces.

\begin{theorem} \label{thm:6-point-1d-coboundary}
There exists a $6$-point $\C_{4, \abs{\cdot}}^1$ space that is not a $\V_4$ space, i.e.~$\C_{4,|\cdot|}\not\subseteq\V_4$.
\end{theorem}
\begin{proof}
Let $(X=\{x_1, \ldots, x_5\}, d)$ be the $5$-point $3$-metric space such that $d(x_{i_1}, x_{i_2}, x_{i_3}) = 1$ for all distinct $i_1, i_2, i_3\in [5]$.
By \cref{lem:all1s-coboundary}, $(X, d)$ is a $\C_{3,\abs{\cdot}}^1$ space.

Let $(X' = X \cup \set{a}, d')$ be the apex extension of $(X, d)$ with apex $a$.
Specifically, (i) $d'(x_{i_1}, x_{i_2}, x_{i_3}, x_{i_4}) = 0$ for all distinct $i_1, i_2, i_3, i_4\in [5]$, and (ii) $d'(x_{i_1}, x_{i_2}, x_{i_3}, a) = 1$ for all distinct $i_1, i_2, i_3\in [5]$ (the metric is zero on $4$-tuples containing non-distinct points).
By \cref{lem:apex_extension_coboundary_to_coboundary}, $(X', d')$ is a $6$-point $\C_{4, \abs{\cdot}}^1$ space. 

We now show that $(X', d')$ is not a $\V_4$ metric. Suppose for contradiction that $(X', d')\in \V_4$, and let $f:X'\rightarrow \R^m$ be the mapping that realizes $d'$.
Because of (i), any four points in $f(x_1), \ldots, f(x_5)$ span a tetrahedron with volume $0$, and so $f(x_1), \ldots, f(x_5)$ all lie in a plane $P$.
Let $h$ be the distance of $f(a)$ from $P$. Because of (ii), the area of any triangle spanned by $f(x_{i_1}), f(x_{i_2}), f(x_{i_3})$, where $i_1, i_2, i_3\in[5]$, must be $3/h$. But, by \cref{lem:unit-triangles-R2}, there can be at most four coplanar points with this property, which is a contradiction.
\end{proof}

\subsubsection{A Lower Bound for Embedding Coboundary \texorpdfstring{$k$}{k}-metrics}
\label{sec:lower-bound-coboundary-volume}

In this subsection, we show that for every triple of distinct points in $X = \set{\vec{x}_1, \ldots, \vec{x}_n} \subseteq \R^m$ to span a triangle of area at least $c_1$ and at most $c_2$ for constants $0 < c_1 \leq c_2$ we must have $m \geq \Omega(\log n)$. From this, we conclude that the discrete $3$-metric, which is a coboundary $3$-metric (\cref{lem:all1s-coboundary}), is not embeddable into a volume $3$-metric space with $m = o(\log n)$ dimensions using constant distortion.

The idea behind the proof is as follows.
By rotation and shift invariance of distance and area, we assume without loss of generality that $\vec{0}, \ell \vec{e}_1 \in X$, where $\ell > 0$ is the maximum distance between two points in $X$, and $\vec{e}_1$ is the unit vector whose first coordinate is $1$ (with all other coordinates equal to $0$). Because each triple of points $\vec{0}, \ell \vec{e}_1, \vec{z}$ for $\vec{z} \in X \setminus \set{\vec{0}, \ell \vec{e}_1}$ spans a triangle of area between $c_1$ and $c_2$, we have that all such $\vec{z}$ lie in a cylinder of radius $2 c_2/\ell$ whose axis is the line segment $(\vec{0}, \ell \vec{e}_1)$.
Then, by applying an appropriate linear transformation $D$ to $X$ (and implicitly the cylinder containing the points in $X$), we get a constant upper bound on the distance between pairs of points in $DX$ and retain a constant lower bound on the area of triangles spanned by distinct triples of points in $DX$. From this, we then additionally get a constant lower bound on $\norm{D(\vec{x} - \vec{y})}$ for distinct $\vec{x}, \vec{y} \in X$, from which we derive an upper bound on $n$ using a packing argument.

We start with an elementary claim relating the lengths of the sides of a triangle to the area of a triangle.

\begin{claim} \label{clm:vol-versus-side-length}
Let $m \in \Z^+$, let $\vec{x}, \vec{y}, \vec{z} \in \R^m$ be distinct, let $T = \conv(\vec{x}, \vec{y}, \vec{z})$, and let \begin{align*}
    \mu^+ :=& \max \set{\norm{\vec{y} - \vec{x}}, \norm{\vec{z} - \vec{x}}, \norm{\vec{z} - \vec{y}}} \ \text{,} \\
    \mu^- :=& \min \set{\norm{\vec{y} - \vec{x}}, \norm{\vec{z} - \vec{x}}, \norm{\vec{z} - \vec{y}}}\ \text{.}
\end{align*}Then $\mu^- \geq 2 \vol_2(T)/\mu^+$.
\end{claim}

\begin{proof}
Assume without loss of generality that $\mu^+ = \norm{\vec{y} - \vec{x}}$ and that $\mu^- = \norm{\vec{z} - \vec{x}}$. Let $M := (\vec{y} - \vec{x}, \vec{z} - \vec{x})$. Then
\[
\vol_2(T) 
= \half \sqrt{\det(M^T M)}
= \half (\norm{\vec{y} - \vec{x}}^2 \norm{\vec{z} - \vec{x}}^2 - \iprod{\vec{y} - \vec{x}, \vec{z} - \vec{x}}^2)^{1/2}
\leq \half \norm{\vec{y} - \vec{x}} \norm{\vec{z} - \vec{x}} = \half \mu^+ \mu^- \ \text{.}
\]
The claim follows by multiplying through by $2/\mu^+$.
\end{proof}

We now prove our main theorem.

\begin{theorem} \label{thm:unit-triangles}
Let $X \subseteq \R^m$ with $n := \card{X}$ satisfying $3 \leq n < \infty$, and suppose that there exist constants $c_1, c_2$ satisfying $0 < c_1 \leq c_2$ such that $c_1 \leq \vol_2(\conv(\vec{x}, \vec{y}, \vec{z})) \leq c_2$ for all distinct triples of points $\vec{x}, \vec{y}, \vec{z} \in X$. Then $m \geq \Omega(\log n)$.
\end{theorem}

\begin{proof}
Let $\ell := \max \norm{\vec{x} - \vec{y}}$, where the maximum is taken over all distinct pairs of points $\vec{x}, \vec{y} \in X$, and assume without loss of generality that $\vec{0}, \ell \vec{e}_1 \in X$.
Let $D := \diag(1/\ell, \ell, \ldots, \ell) \in \R^{m \times m}$. 
Let $\pi_1$ denote projection onto $\lspan(\vec{e}_1)$ and let $\pi_1^{\perp}$ denote projection onto $\lspan(\vec{e}_1)^{\perp} = \lspan(\vec{e}_2, \ldots, \vec{e}_m)$.

We claim that for all $\vec{x} \in X$, $\norm{\pi_1^{\perp}(\vec{x})} \leq 2 c_2/\ell$. The claim is clearly true for $\vec{x} = \vec{0}$ and $\vec{x} = \ell \vec{e}_1$, and so assume $\vec{x} \in X \setminus \set{\vec{0}, \ell \vec{e}_1}$. Then using the formula for the area of a triangle, $\vol_2(\conv(\vec{0}, \ell \vec{e}_1, \vec{x})) = \ell/2 \cdot \norm{\pi_1^{\perp}(\vec{x})} \leq c_2$, from which the claim follows by rearranging.

Then for any $\vec{x}, \vec{y} \in X$,
\begin{align} \label{eq:Dxy-dist-ub}
\begin{split}
\norm{D(\vec{x} - \vec{y})}^2 &= \norm{\pi_1(D(\vec{x} - \vec{y}))}^2 + \norm{\pi_1^{\perp}(D(\vec{x} - \vec{y}))}^2 \\
&= 1/\ell^2 \cdot \norm{\pi_1(\vec{x} - \vec{y})}^2 + \ell^2 \cdot \norm{\pi_1^{\perp}(\vec{x} - \vec{y})}^2 \\
&\leq 1/\ell^2 \cdot \ell^2 + \ell^2 \cdot \norm{\pi_1^{\perp}(\vec{x} - \vec{y})}^2 \\
&\leq 1/\ell^2 \cdot \ell^2 + \ell^2 \cdot (4 c_2/\ell)^2 \\
&= 16 (c_2)^2 + 1 \ \text{.}
\end{split}
\end{align}

The first inequality uses $\norm{\pi_1(\vec{x} - \vec{y})} \leq \norm{\vec{x} - \vec{y}} \leq \ell$, which follows from the definition of $\ell$. The second inequality uses triangle inequality and the claim above, which implies that  $\norm{\pi_1^{\perp}(\vec{x})} \leq 2 c_2/\ell$ and  $\norm{\pi_1^{\perp}(\vec{y})} \leq 2 c_2/\ell$.

Now, let $\vec{x}, \vec{y}, \vec{z} \in X$, and let $T := \conv(\vec{x}, \vec{y}, \vec{z})$. Then
\begin{equation} \label{eq:vol-DT-lb}
    \vol_2(DT)^2 = \sum_{\substack{I \subseteq [m], \\ \card{I} = 2}} \vol_2(\pi_I(DT))^2 
    = \sum_{\substack{I \subseteq [m], \\ \card{I} = 2}} \vol_2(D_{I, I} \cdot \pi_I(T))^2
    \geq \sum_{\substack{I \subseteq [m], \\ \card{I} = 2}} \vol_2(\pi_I(T))^2
    = \vol_2(T)^2 
    \geq (c_1)^2 \ \text{,}
\end{equation}
where the first and third equalities use the Cauchy-Binet formula (\cref{thm:cauchy-binet}), and the inequality uses the determinant formula for volume (\cref{eq:volume}) and the fact that $\det(D_{I, I}) \geq \min \set{1/\ell \cdot \ell, \ell^2} = 1$ for all $I \subseteq [m]$ with $\card{I} = 2$.

Combining \cref{eq:Dxy-dist-ub,eq:vol-DT-lb,clm:vol-versus-side-length}, we get that for all $\vec{x}, \vec{y} \in X$ with $\vec{x} \neq \vec{y}$, $\norm{D(\vec{x} - \vec{y})} \geq C$, where
\[
C := 2 c_1/\sqrt{16 (c_2)^2 + 1} \ \text{.}
\]
Therefore, the scaled, shifted Euclidean balls $(C/2 \cdot \B_2^m) + D \vec{x}$, $(C/2 \cdot \B_2^m) + D \vec{y}$ for such distinct $\vec{x}, \vec{y}$ are interior disjoint, and, again using \cref{eq:Dxy-dist-ub}, we have that all balls $(C/2 \cdot \B_2^m) + D \vec{x}$ for $\vec{x} \in X$ are contained in a larger Euclidean ball of radius $1/2 \cdot (\sqrt{16 (c_2)^2 + 1} + C)$.\footnote{Here $\B_2^m$ is the closed unit Euclidean ball in $m$ dimensions.}
It follows that the number $n$ of balls of the form $(C/2 \cdot \B_2^m) + D \vec{x}$ and hence points in $X$ is bounded as 
\[
n \leq \frac{\vol_m(1/2 \cdot (\sqrt{16 (c_2)^2 + 1} + C) \cdot \B_2^m)}{\vol_m(C/2 \cdot \B_2^m)} 
= (1/C \cdot \sqrt{16 (c_2)^2 + 1} + 1)^m
= \Big( \frac{16 (c_2)^2 + 1}{2 c_1} + 1 \Big)^m \ \text{,}
\]
and the claim follows by taking logs.
\end{proof}

We then get the following corollary, which says that the all-ones metric in \cref{lem:all1s-coboundary}, which is a coboundary $3$-metric, is not embeddable as a $o(\log n)$-dimensional volume $3$-metric with constant distortion.

\begin{corollary} \label{cor:one-cbd-many-vol}
There exists a $\C_{3, |\cdot|}^1$ coboundary metric space $(X, d)$ with $n = \card{X}$ vertices with the following property. If there exist a space $(X', d')$ in $\V_3^m$, a function $g: X \to X'$, and constants $c_1, c_2$ satisfying $0 < c_1 \leq c_2$ such that for all $x, y, z \in X$, $c_1 \cdot d(x, y, z) \leq d'(g(x), g(y), g(z)) \leq c_2 \cdot d(x, y, z) $ then $m \geq \Omega(\log n)$.
\end{corollary}

\begin{proof}
Combine \cref{lem:all1s-coboundary} and \cref{thm:unit-triangles}.
\end{proof}

We remark that a substantial amount of work has gone into bounding the maximum possible number of unit-area triangles spanned by $n$ points in fixed low-dimensional space (in $\R^2$ or $\R^3$)~\cite{ERDOS1971246,journals/jct/DumitrescuST09,journals/dcg/ApfelbaumS10,journals/combinatorica/RazS17}.
The special case of \cref{thm:unit-triangles} where $c_1 = c_2 = 1$ gives an asymptotic lower bound on $m$ as a function of $n$ and to isometrically embedding the all-ones $3$-metric in \cref{lem:all1s-coboundary} as a volume metric.%
\footnote{We note that it is \emph{possible} for all triples of $n$ points to span unit-area triangles when $m$ is sufficiently large. To see this, consider the triangles spanned by triples of (scalings of) the standard-normal basis vectors $\vec{e}_1, \ldots, \vec{e}_n \in \R^n$. By shifting by $-\vec{e}_1$ and applying dimension reduction, it is possible to use $m = n-1$ dimensions, which one might conjecture is optimal. However, in general it is not. As noted by \cref{lem:unit-triangles-R2}, it is possible to achieve $n = 4$ in $m = 2$ dimensions by taking the vertices of a square.}

\bibliographystyle{alpha}
\bibliography{bibliographies/distortion,bibliographies/k-metrics}

\end{document}